\documentclass{article}
\usepackage{paralist}

\usepackage{amsthm}     

\newtheorem{theorem}{Theorem}

\newtheorem{lemma}{Lemma}

\newtheorem{remark}{Remark}

\usepackage[english]{babel}
\usepackage[utf8]{inputenc} 
\usepackage{amsfonts}
\usepackage{amssymb,amsmath} 
\usepackage{graphics,graphicx} 
\usepackage{color}
\usepackage{colordvi}
\usepackage{gastex}
\usepackage{comment} 
\usepackage{enumerate}
\renewcommand{\sb}[1]{\scalebox{0.75}[1]{#1}}
\newcommand{\sbcustom}[2]{\scalebox{1}[1]{#2}}    

\makeatletter
\DeclareRobustCommand\sfrac[1]{\@ifnextchar/{\@sfrac{#1}}%
                                            {\@sfrac{#1}/}}
\def\@sfrac#1/#2{\leavevmode\scalebox{.9}{\kern.1em\raise.5ex
         \hbox{$\m@th\mbox{\fontsize\sf@size\z@
                           \selectfont#1}$}\kern-.1em
         /\kern-.15em\lower.25ex
          \hbox{$\m@th\mbox{\fontsize\sf@size\z@
                            \selectfont#2}$}}}
\DeclareRobustCommand\numfrac[1]{\@ifnextchar/{\@numfrac{#1}}%
                                            {\@numfrac{#1}}}
\def\@numfrac#1{\leavevmode \hbox{$\m@th\mbox{\fontsize\sf@size\z@
                           \selectfont#1}$}}
\makeatother

\newcommand{\real}{\mathbb R}
\newcommand{\rat}{\mathbb Q}

\newcommand{\tuple}[1]{\langle #1 \rangle}
\def\abs#1{\ensuremath{\lvert #1\rvert}} 

\let\epsilon\varepsilon
\let\emptyset\varnothing

\newcommand{\Inf}{{\sf Inf}}
\newcommand{\cale}{\mathcal E}
\newcommand{\Supp}{{\sf Supp}}
\newcommand{\Last}{\mathsf{Last}}

\newcommand{\half}{$\frac{\text{1}}{\text{2}}$}
\newcommand{\DD}{\Delta}
\newcommand{\GG}{\mathcal{G}}
\newcommand{\PP}{\delta}

\newcommand{\winval}[1]{\langle \! \langle #1 \rangle\! \rangle_{\mathit{val}} }
\newcommand{\winvalf}[1]{\winval{#1}^{{FM}}}
\newcommand{\va}{\winval{\ma}}
\newcommand{\vb}{\winval{\mi}}
\newcommand{\vaf}{\winval{\ma}^{{FM}}}
\newcommand{\vbf}{\winval{\mi}^{{FM}}}

\newcommand{\straava}{\winval{\straa}}
\newcommand{\straavaf}{\winval{\straa}^{{FM}}}

\newcommand{\strabvbf}{\winval{\strab}^{{FM}}}
\newcommand{\Win}{\mathsf{Win}}
\newcommand{\mystravaf}[1]{\winval{#1}^{{FM}}}

\newcommand{\cpremi}{\mathsf{Cpre}_{\mi}}

\newcommand{\atmi}{\mathsf{Attr}_{\mi}}
\newcommand{\straa}{\sigma}
\newcommand{\Straa}{\Sigma}
\newcommand{\strab}{\pi}
\newcommand{\Strab}{\Pi}

\newcommand{\LP}{{\sf LP}}

\DeclareMathOperator{\ma}{\mathsf{Max}}
\DeclareMathOperator{\mi}{\mathsf{Min}}
\newcommand{\prob}[1]{\mathbb{P}_{#1}}
\newcommand{\rwd}{\mathsf{rwd}}

\newcommand{\MeanSup}{\operatorname{\mathsf{MeanSup}}}
\newcommand{\MeanInf}{\operatorname{\mathsf{MeanInf}}}
\newcommand{\Buchi}{\operatorname{\mathsf{B\ddot{u}chi}}}
\newcommand{\coBuchi}{\operatorname{\mathsf{coB\ddot{u}chi}}}

\begin{document}

\sloppy

\title{{\bf Perfect-Information Stochastic Games \mbox{with Generalized Mean-Payoff Objectives}}\thanks{This research was partially supported by Austrian Science Fund (FWF) 
NFN Grant No S11407-N23 (RiSE/SHiNE), ERC Start grant (279307: Graph Games), 
Vienna Science and Technology Fund (WWTF) through project ICT15-003,
and European project Cassting (FP7-601148).
}}

\author{
Krishnendu Chatterjee$^\dag$ \quad  Laurent Doyen$^{\S}$ \\ 
\normalsize
 $\strut^\dag$ IST Austria \quad $\strut^\S$ CNRS \& LSV, ENS Cachan 
}

\date{}
\maketitle

\begin{abstract}
Graph games provide the foundation for modeling and synthesizing 
reactive processes. In the synthesis of stochastic reactive processes, 
the traditional model is perfect-information stochastic games, 
where some transitions of the game graph are controlled by two adversarial 
players, and the other transitions are executed probabilistically. 
We consider such games where the objective is the conjunction of 
several quantitative objectives (specified as mean-payoff conditions),
which we refer to as generalized mean-payoff objectives.
The basic decision problem asks for the existence of a 
finite-memory strategy for a player that ensures the 
generalized mean-payoff objective be satisfied with a desired
probability against all strategies of the opponent.
A special case of the decision problem is the almost-sure
problem where the desired probability is~1. 
Previous results presented a semi-decision procedure for $\epsilon$-approximations 
of the almost-sure problem. 
In this work, we show that both the almost-sure problem as well
as the general basic decision problem are coNP-complete,
significantly improving the previous results.
Moreover, we show that in the case of 1-player stochastic
games, randomized memoryless strategies are sufficient and the problem 
can be solved in polynomial time.
In contrast, in two-player stochastic games, we show that 
even with randomized strategies exponential memory is required in general,
and present a matching exponential upper bound.
We also study the basic decision problem with infinite-memory strategies
and present computational complexity results for the problem. 
Our results are relevant in the synthesis of stochastic reactive 
systems with multiple quantitative requirements.
\end{abstract}

\section{Introduction}
\label{sec-intro}

Reactive systems are non-terminating processes that interact continually with a changing environment. 
Since such systems are non-terminating, their behavior is described by infinite sequences of events.
The classical framework to model reactive systems with controllable and uncontrollable events are games on graphs. 
In the presence of uncertainties, we have stochastic reactive systems with probability distributions over state changes. 
The performance requirement on such systems, such as power consumption or latency,
can be represented by rewards (or costs) associated to the events of the system,
and a quantitative objective that aggregates the rewards of an execution to a single value.
In several modeling domains, however, there is not a single objective 
to be optimized, but multiple, potentially dependent and conflicting goals.
For example, in the design of an embedded system, the goal may be to maximize 
average performance while minimizing average power consumption.
Similarly, in an inventory management system, the goal would be to optimize
the costs associated to maintaining each kind of product~\cite{FV97,AltmanBook}. 
Thus it is relevant to study stochastic games with multiple 
quantitative objectives.

\smallskip\noindent{\em Perfect-information stochastic games.}
A perfect-information stochastic graph game~\cite{Condon92}, also known as turn-based stochastic 
game or \emph{2\half-player graph game}, consists of a finite directed graph with three kinds of states (or vertices): 
player-$\ma$, player-$\mi$, and probabilistic states.
The game starts at an initial state, and is played as follows: 
at player-$\ma$ states, player~$\ma$ chooses a successor state;
at player-$\mi$ states, player~$\mi$ (the adversary of player~$\ma$)  does likewise;
and at probabilistic states, a successor state 
is chosen according to a fixed probability distribution.
Thus the result of playing the game forever is an infinite path through 
the graph.
If there are no probabilistic states, we refer to the game as a 
\emph{2-player graph game}; 
if there are no player-$\mi$ states, we refer to the (1\half-player) game
as a Markov decision process (MDP);
if there are no probabilistic states and no player-$\mi$ states,
then the (1-player) game is a standard graph.

The class of 2-player graph games has been used for a long time 
to synthesize non-stochastic reactive systems~\cite{BL69,PnueliRosner,RamadgeWonham}:
a reactive system and its environment represent the two players, whose states 
and transitions are specified by the vertices and edges of a game graph.
Similarly, MDPs have been used to model stochastic processes without 
adversary~\cite{FV97,Puterman}.
Consequently, 2\half-player graph games, which subsume both 2-player 
graph games and MDPs, provide the theoretical foundation
to model stochastic reactive systems~\cite{FV97,RF91}.

\smallskip\noindent{\em Mean-payoff objectives.}
One of the most classical example of quantitative objectives is the mean-payoff 
objective~\cite{FV97,Puterman,Gi57,EM79}, where a reward is associated to each
state and the payoff of a path is the long-run average of the rewards of the 
path (computed as either $\liminf$ or $\limsup$ of the averages of the finite 
prefixes to ensure the payoff value always exists).
While traditionally the verification and the synthesis problems were considered
with Boolean objectives~\cite{PnueliRosner,RamadgeWonham,KV05}, recently quantitative objectives have 
received a lot of attention~\cite{BCHJ09,CCHRS11,BBFR13}, as they specify requirements 
on resource consumption (such as for embedded systems or power-limited systems) as well 
as performance-related properties.

\smallskip\noindent{\em Various semantics for multiple quantitative objectives.}
The two classical semantics for quantitative objectives are as follows~\cite{BBCFK14}:
the first is the expectation semantics, which is a probabilistic average of 
the quantitative objective over the executions of the system;
and the second is the satisfaction semantics, which consider the probability of
the set of executions where the quantitative objective is at least a required
threshold value~$\nu$.
The expectation objective is relevant in situations where we are interested 
in the ``average'' behaviour of many instances of a given system, while 
the satisfaction objective is useful for analyzing and optimizing the desired
executions, and is more relevant for the design of critical stochastic reactive 
systems (see~\cite{BBCFK14} for a more detailed discussion).
For example, consider one mean-payoff objective that specifies the set of 
executions where the average power consumption is at most 5 units, and 
another mean-payoff objective that specifies the set of executions where 
the average latency is at most 10 units. 
A multiple objective asks to \emph{satisfy} both, i.e., their conjunction. 
We refer to such objectives (i.e., conjunction of multiple mean-payoff 
objectives) as {\em generalized mean-payoff objectives}\footnote{In the verification
literature, conjunction of reachability, B\"uchi, and parity objectives, 
are referred to as generalized reachability, generalized B\"uchi, and generalized
parity objectives, respectively, and generalized mean-payoff objectives naming is 
for consistency.}. 
The goal of player~$\ma$ is to maximize the probability of satisfaction 
of the generalized mean-payoff objective while player~$\mi$ tries to 
minimize this probability, i.e., the game is zero-sum.
Concrete applications of 
2\half-player graph games with generalized mean-payoff objectives have been considered, such as 
best-effort synthesis where the goal is to minimize the violation 
of several incompatible specifications~\cite{CGHRT12}, real-time
scheduling algorithms with requirements on the utility and energy consumption~\cite{CPKS14},
and electric power distribution in an avionics application~\cite{BKTW15}.
In particular, for the real-world avionics application in~\cite{BKTW15},
both two adversarial players, stochastic transitions, as well as multiple mean-payoff
objectives are required, i.e., the application can be modeled as 2\half-player
graph games with generalized mean-payoff objectives, but not in a strict subclass.

\smallskip\noindent{\em Computational questions.} 
In this work, we consider 2\half-player graph games with generalized 
mean-payoff objectives in the satisfaction semantics.
A strategy for a player is a recipe that given the history of interaction so far 
(i.e., the sequence of states) prescribes the next move.
The basic decision problem asks, given a 2\half-player graph 
game, a generalized mean-payoff objective, and a probability threshold 
$\alpha$, whether there exists a strategy for player~$\ma$ to ensure the objective 
be satisfied with probability at least $\alpha$ against all strategies of 
player~$\mi$. 
Since strategies in games correspond to implementations of controllers for 
reactive systems, a particularly relevant question is to ask for the existence 
of a finite-memory strategy in the basic decision problem, instead of an 
arbitrary strategy.
Moreover, an important special case of the basic decision problem is the 
almost-sure problem, where the probability threshold $\alpha$ is equal to~$1$.

\smallskip\noindent{\em Previous results.}
We summarize the main previous results for MDPs, 2-player graph games, and 
2\half-player graph games, with generalized mean-payoff objectives.

\begin{compactenum}
\item {\em MDPs.} The basic decision problem for generalized mean-payoff objectives in MDPs 
with infinite-memory strategies can be solved in polynomial time~\cite{BBCFK14}. 
The problem under finite-memory strategies has not been addressed yet.

\item {\em 2-player games.} 
The following results are known~\cite{VCDHRR15}:
the basic decision problem for generalized mean-payoff objectives in 
2-player graph games, both under finite-memory and infinite-memory 
strategies, is coNP-complete;
moreover, for infinite-memory strategies if the mean-payoff objective
is defined as the limit supremum of the averages (rather than limit infimum of 
the average), then the problem is in NP $\cap$ coNP.

\item {\em 2\half-player games.} 
The almost-sure problem for generalized mean-payoff objectives in 2\half-player 
graph games under finite-memory strategies was considered in~\cite{BKTW15}, and 
a semi-algorithm (or semi-decision procedure) was presented for approximations
of the problem.

\item {\em Memory of strategies}. Infinite-memory strategies are strictly more powerful than finite-memory 
strategies, even in 1-player graph games thus also in MDPs and 2-player graph games:
there are games where an infinite-memory strategy can ensure the objective with probability~1 
while all finite-memory strategies fail to do so\footnote{However, in some variants of the decision problem 
(such as requiring the mean-payoff value, computed as the $\liminf$ of the averages of the finite prefixes, 
be strictly greater than a threshold~$\nu$) finite-memory strategies are as powerful as
infinite-memory strategies in MDPs~\cite{CR15}. }~\cite{VCDHRR15}.
\end{compactenum}

\smallskip\noindent{\em Our contributions.}
The previous results suggest that 2\half-player graph games with generalized 
mean-payoff objectives are considerably more complicated than 2-player graph 
games as well as MDPs, as even the decidability of the almost-sure problem was 
open for 2\half-player graph games for finite-memory strategies (the previous 
result neither gives an exact algorithm, nor establishes decidability for 
approximation).
In this work we present a complete picture of decidability as well as computational 
complexity. 
Our results are as follows:
\begin{compactenum}

\item {\em MDPs.} First we study the generalized mean-payoff problem under 
finite-memory strategies in MDPs. We present a polynomial-time algorithm,
and show that with randomization, memoryless strategies (which do not depend on
histories but only on the current state) are sufficient,
i.e., for finite-memory optimal strategies no memory is required.

\item {\em 2\half-player games.} 
For 2\half-player graph games with generalized mean-payoff objectives we show 
that: 
(1)~the basic decision problem is coNP-complete under finite-memory strategies (significantly improving 
the known semi-decidability result for approximation of the almost-sure 
problem~\cite{BKTW15}), and moreover, the same complexity holds for the almost-sure problem; 
and (2)~under infinite-memory strategies, the computational complexity results 
coincide with the special case of 2-player graph games.

\item {\em Memory of strategies.} Under finite-memory strategies, 
in contrast to MDPs where we show with randomization no memory is required,
we establish an exponential lower bound (even with randomization) for memory 
required in 2\half-player graph games with generalized mean-payoff objectives.
We also present a matching upper bound showing that exponential memory
is sufficient.
\end{compactenum}

\smallskip\noindent{\em Key technical insights.}
We show that for generalized mean-payoff objectives, for the adversary, 
pure and memoryless strategies are sufficient. 
Under finite-memory strategies for player~$\ma$, this result is established
using the following ideas: 
\begin{compactitem}

\item In general for prefix-independent objectives (objectives that do not change
if finite prefixes are added or removed from a path), we show that sub-game perfect 
strategies exist, where a strategy is sub-game perfect if it is optimal after every
finite history.
Such a result is known for infinite-memory strategies using results
from martingale theory~\cite{GK14}. 
Our proof for finite-memory strategies is conceptually simpler, and 
uses combinatorial arguments and well-known discrete properties of MDPs 
(see Lemma~\ref{lem:subgame-perfect}, Section~\ref{sec:half}).

\item Then using the above result we show that for a sub-class of 
prefix-independent objectives (that subsume generalized mean-payoff objectives) 
for the adversary pure memoryless strategies suffice 
(see Theorem~\ref{theo:half-memoryless}, Section~\ref{sec:half}).
Moreover, for this class of objectives we establish determinacy when each
player is restricted to finite-memory strategies, which is of independent 
interest (see also Theorem~\ref{theo:half-memoryless}); and also show that such 
determinacy result does not hold for all prefix-independent objectives 
(see Remark~\ref{rem:nodet}).

\item For MDPs, we generalize a result of~\cite{KS88} from graphs to MDPs, 
to obtain a linear-programming solution for the generalized mean-payoff objectives 
under finite-memory strategies (see Theorem~\ref{theo:mdp}, Section~\ref{sec:finmem}). 
\end{compactitem}
Combining these results we obtain the coNP upper bound for the basic 
decision problem for 2\half-player graph games and the coNP 
lower bound follows from existing results on 2-player graph games 
(see Theorem~\ref{theo:finite-coNP-complete}, Section~\ref{sec:finmem}).

\smallskip\noindent{\em Related works.} 
We have described the most relevant related works in the paragraph {\em Previous
results.} 
We discuss other relevant related works.
Markov decision processes with multiple objectives have been studied in numerous
works, for various quantitative objectives, such as mean-payoff~\cite{Cha07,BBCFK14}, 
discounted sum~\cite{CMH06,CFW13}, total reward~\cite{FKN11} as well as qualitative 
objectives~\cite{EKVY08}, and 
their combinations~\cite{CKK15,Baier-CSL-LICS-1,Baier-CSL-LICS-2,CR15}.
The problem of 2-player graph games with multiple quantitative objectives has 
also been widely studied both for finite-memory strategies~\cite{VCDHRR15,CRR14,JLS15,BR15,Vel14}
as well as infinite-memory strategies~\cite{VCDHRR15,CV13}.
In contrast, for 2\half-player games with multiple quantitative objectives only
few results are known~\cite{BKTW15,CFKSW13}, because of the inherent difficulty to handle two-players, probabilistic
transitions, as well as multiple objectives all at the same time.
A semi-decision procedure for approximation of the almost-sure problem for 2\half-player
games with generalized mean-payoff objectives was presented in~\cite{BKTW15}, 
which we significantly improve.
The class of 2\half-player graph games with positive Boolean combinations of total-reward
objectives was considered in~\cite{CFKSW13}, and the problem was established to be
PSPACE-hard and undecidable for pure strategies.


\section{Definitions}\label{sec:def}


\noindent{\bf Probability distributions.}
For a finite set $S$, we denote by $\DD(S)$ the set of all probability 
distributions over $S$, i.e., the set of functions $p: S\to [0,1]$ 
such that $\sum_{s\in S} p(s) = 1$.
The \emph{support} of $p$ is the set $\Supp(p) = \{s \in S \mid p(s) > 0\}$.
For a set $U \subseteq S$ let $p(U) = \sum_{s \in U} p(s)$.

\smallskip\noindent{\bf Perfect-information stochastic games.}
A \emph{perfect-information stochastic game} (for brevity, stochastic games
in the sequel) 
is a tuple $\GG = \tuple{S, (S_{\ma}, S_{\mi}), A, \PP}$, 
consisting of a finite set $S = S_{\ma} \uplus S_{\mi}$ of states
partitioned into the set $S_{\ma}$ of states controlled by player~$\ma$ (depicted
as round states in figures) and the set $S_{\mi}$ of states controlled by player~$\mi$ (depicted
as square states in figures), a finite set $A$ of actions, 
and a probabilistic transition function $\PP: S \times A \to \DD(S)$.
If $\PP(s,a)(s') > 0$, we say that $s'$ is an \emph{$a$-successor} of $s$.
A transition $\PP(s,a)$ is \emph{deterministic} if $\PP(s,a)(s') = 1$ for some 
state $s'$. 
The underlying graph of $\GG$ is $(S, E)$ where $E = \{(s,s') \mid 
\PP(s,a)(s') > 0 \text{ for some } a \in A \}$.

For complexity results, we consider that the probabilities in stochastic games
are rational numbers with numerator and denominator 
encoded in binary.

\smallskip\noindent{\bf Markov decision processes and end-components.}
A \emph{Markov decision process} (MDP) is the special case of a stochastic
game where either $S_{\ma} = \emptyset$, or $S_{\mi} = \emptyset$. 
Given a state $s \in S$ and a set $U \subseteq S$,
let $A_U(s)$ be the set of all actions $a \in A$ such 
that $\Supp(\PP(s,a)) \subseteq U$. 
A \emph{closed} set in an MDP is a set $U \subseteq S$ such 
that $A_U(s) \neq \emptyset$ for all $s \in U$. 
%
A set $U \subseteq S$ is an {\em end-component}~\cite{CY95} if (i)~$U$ is closed, and 
(ii)~the graph $(U,E_U)$ is strongly connected
where $E_U=\{(s,t) \in U \times U \mid \PP(s,a)(t) > 0 \text{ for some } a \in A_U(s)\}$
denote the set of edges given the actions.
We denote by $\cale(M)$ the set of all end-components of an MDP~$M$.  

\smallskip\noindent{\bf Markov chains and recurrent sets.}
A \emph{Markov chain} is the special case of an MDP 
where the action set $A$ is a singleton. In Markov chains,
end-components are called \emph{closed recurrent sets}.


\smallskip\noindent{\bf Plays and strategies.}
A \emph{play} is an infinite sequence $s_0 s_1 \ldots \in S^\omega$ of states.
A \emph{randomized strategy} for $\ma$ is a recipe to describe what is the next action to play
after a prefix of a play ending in a state controlled by player~$\ma$;
formally, it is a function $\straa:S^*S_{\ma} \to \DD(A)$ that provides
probability distributions over the action set. 
A \emph{pure strategy} is a function $\straa:S^*S_{\ma} \to A$ that provides
a single action, which can be seen as a special case of randomized strategy
where for every play prefix $\rho \in S^*S_{\ma}$ there exists an action $a \in A$ 
such that $\straa(\rho)(a) = 1$.

We consider the following memory restrictions on strategies.
A strategy $\straa$ is \emph{memoryless} if it is independent of the past and depends only on 
the current state, that is $\straa(\rho) = \straa(\Last(\rho))$ for all play prefixes 
$\rho \in S^*S_{\ma}$, where $\Last(s_0 \dots s_k) = s_k$.
In the sequel, we call memoryless strategies the pure memoryless strategies,
and we emphasize that strategies $\straa: S_{\ma} \to \DD(A)$ are not necessarily 
pure by calling them randomized memoryless.


A strategy $\straa$ uses \emph{finite memory} if it can be described 
by a transducer $\tuple{M, m_0, \straa_u, \straa_n}$ consisting of a finite
set $M$ (the memory set), an initial memory value $m_0 \in M$, 
an update function $\straa_u: M \times S \to M$ for the memory, 
and a next-action function $\straa_n: M \to \DD(A)$; the transducer $\tuple{M, m_0, \straa_u, \straa_n}$ 
defines the strategy $\straa$ such that $\straa(\rho) = \straa_n(\hat{\straa}_u(m_0,\rho))$ 
for all play prefixes $\rho \in S^*S_{\ma}$
where $\hat{\straa}_u$ extends $\straa_u$ to sequences of states as usual (i.e.,
$\hat{\straa}_u(m, \rho \cdot s) = \straa_u(\hat{\straa}_u(m, \rho), s)$).
Given a finite-memory strategy $\straa$ for player~$\ma$, let $\GG_\straa = \tuple{S', (\emptyset, S'_{\mi}), A, \PP'}$ 
be the MDP obtained by playing $\straa$ in $\GG$, where 
$S' = S'_{\mi} = S \times M$    
and the transition function $\PP'$ is defined
for all $\tuple{s,m} \in S'$ and action $a \in A$ of player~$\mi$ as follows,
for all $s'\in S$, where $m' = \straa_u(m,s)$:
\begin{itemize}
\item if $s \in S_{\ma}$, then 
$\PP'(\tuple{s,m}, a)(\tuple{s',m'}) = \sum_{b \in A} \straa_n(m')(b) \cdot \PP(s,b)(s')$; 
\item if $s \in S_{\mi}$, then 
$\PP'(\tuple{s,m}, a)(\tuple{s',m'}) = \PP(s,a)(s')$.
\end{itemize}

Strategies $\strab$ for player~$\mi$ are defined analogously, as well as the memory restrictions.
A strategy that is not finite-memory is referred to as an infinite-memory 
strategy. We denote by $\Straa$ the set of all strategies for player~$\ma$, 
and by $\Straa^{PM}$, 
and $\Straa^{FM}$ respectively the set of all pure memoryless,
and all finite-memory strategies 
for player~$\ma$. We use analogous notation $\Strab$, $\Strab^{PM}$, 
and $\Strab^{FM}$ for player~$\mi$.

\smallskip\noindent{\bf Objectives.}
An {\em objective} is a Borel-measurable set of plays~\cite{Billingsley}.
In this work we consider conjunctions of mean-payoff objectives. 
Some of our results are related to more general classes of prefix-independent
and shuffle-closed objectives.
We define the relevant objectives below:
\smallskip
\begin{compactenum}
\item {\em Prefix-independent objectives.}
An objective $\Omega \subseteq S^\omega$ is \emph{prefix-independent} if for all plays $\rho \in S^\omega$,
and all states $s \in S$, we have $\rho \in \Omega$ if and only if $s \cdot \rho \in \Omega$,
that is the objective is independent of the finite prefixes (of arbitrary length) 
of the plays.

\item {\em Shuffle-closed objectives.} 
A \emph{shuffle} of two plays $\rho_1$, $\rho_2$ is a play $\rho = u_1 u_2 u_3 \dots$
such that $u_i \in S^*$  for all $i \geq 1$, and 
$\rho_1 = u_1 u_3 u_5 \dots$ and $\rho_2 = u_2 u_4 u_6 \dots$.
An objective $\Omega \in S^{\omega}$ is closed under shuffling, 
if all shuffles of all plays $\rho_1, \rho_2 \in \Omega$ belong to $\Omega$.

\item {\em Multi-mean-payoff objectives.}
Let $\rwd: S \to \rat^k$ be a \emph{reward function}\footnote{We use rational
rewards to be able to state complexity results. All other results in this paper
hold if the rewards are real numbers.} 
that assigns a $k$-dimensional vector of weights to each state. 
For $1 \leq j \leq k$, we denote by $\rwd_j: S \to \rat$ the projection of 
the function $\rwd$ on the $j$-th dimension.
The conjunction of \emph{mean-payoff-inf} objectives (which we refer as generalized 
mean-payoff objectives) is the set  
$$
\sbcustom{.95}{$\displaystyle
\MeanInf=
 \bigg\{
  s_0s_1\cdots \in S^\omega
  \mid 
  \bigwedge_{j=1}^{k} \liminf_{n \to \infty} \frac{1}{n} \cdot \sum_{i = 0}^{n-1}\rwd_j(s_{i}) \geq 0 
  \bigg\}
$}
$$
that contains all plays for which the long-run average of weights (computed as $\liminf$)
is non-negative\footnote{Note that it is not restrictive to define mean-payoff objectives with a threshold $0$
since we can obtain mean-payoff objectives defined as the long-run average of weights 
above any threshold $\nu$ by subtracting the constant $\nu$ to the reward function.} 
in all dimensions. 
The objectives inside the above conjunction (indexed by $j$) are called 
one-dimensional mean-payoff-inf objectives (in dimension $j$), and denoted $\MeanInf_j$.
The conjunction of \emph{mean-payoff-sup} objectives is the set $\MeanSup$ 
defined analogously, replacing $\liminf$ by $\limsup$ in the definition of $\MeanInf$.
\end{compactenum}


\begin{remark}\label{rmk:mean-payoff-inf-not-closed-under-shuffling}
It is easy to show that mean-payoff-inf objectives are closed under shuffling,
and that the conjunction of objectives that are closed under shuffling
is closed under shuffling~\cite{Kop06}. However, the conjunctions of mean-payoff-sup
objectives are in general not closed under shuffling~\cite[Example~1]{VCDHRR15}.
\end{remark}

\smallskip\noindent{\bf Probability measures.}
Given an initial state $s$, and a pair of strategies $(\straa,\strab)$ for $\ma$ and $\mi$,
a finite prefix $\rho = s_0  \cdots s_n$ of a play is \emph{compatible}
with $\straa$ and $\strab$ if $s_0 = s$ and for all $0 \leq i \leq n-1$, there exists an action
$a_i \in A$ such that $\PP(s_{i}, a_i)(s_{i+1}) > 0$, 
and either $s_i \in S_{\ma}$ and $\straa(s_0 \cdots s_{i})(a_i) > 0$,
or $s_i \in S_{\mi}$ and $\strab(s_0 \cdots s_{i})(a_i) > 0$.
A probability can be assigned in a standard way to every finite play prefix $\rho$, 
and by Caratheodary's extension theorem a probability measure $\prob{s}^{\straa,\strab}(\cdot)$ 
of objectives can be uniquely defined.  
For MDPs, we omit the strategy of the player with empty set of states, and 
for instance if $S_{\mi} = \emptyset$ we denote by $\prob{s}^{\straa}(\cdot)$ 
the probability measure under strategy $\straa$ of player~$\ma$.

\smallskip\noindent{\bf Value and almost-sure winning.}
The optimal \emph{value} from an initial state $s$ of a game with objective $\Omega$ is defined by
$$\va(\Omega,s) = \sup_{\straa \in \Straa} \inf_{\strab \in \Strab} \prob{s}^{\straa,\strab}(\Omega).$$
By Martin's determinacy result~\cite{Mar98}, the optimal value is also 
$\vb(\Omega,s) = \inf_{\strab \in \Strab} \sup_{\straa \in \Straa} \prob{s}^{\straa,\strab}(\Omega)$,
the infimum probability of satisfying $\Omega$ that player~$\mi$ can ensure
against all strategies of player~$\ma$.
In other words the determinacy shows that $\va(\Omega,s) = \vb(\Omega,s)$, and the order
of sup and inf in the quantification of the strategies can be exchanged.

A strategy $\straa$ for player~$\ma$ is \emph{optimal} from a state $s$
if for all strategies $\strab$ for player~$\mi$ it ensures that 
$\prob{s}^{\straa,\strab}(\Omega) \geq \va(\Omega,s)$. The value (or winning probability) of
a strategy $\straa$ in state $s$ is $\straava(\Omega,s) = \inf_{\strab \in \Strab} \prob{s}^{\straa,\strab}(\Omega)$. 
We omit analogous definitions for player~$\mi$.

We say that player~$\ma$ wins almost-surely 
from an initial state~$s$ 
if there exists a strategy $\straa$ for $\ma$ such that for every strategy
$\strab$ of player~$\mi$ we have $\mathbb{P}^{\straa,\strab}_{s}(\Omega)=1$.
The state $s$ and the strategy $\straa$ are called \emph{almost-sure} 
winning for player~$\ma$.

\smallskip\noindent{\bf Finite-memory values and almost-sure winning.}
The optimal \emph{finite-memory value} (for player~$\ma$) is defined analogously, 
when the players are restricted to finite-memory strategies:
$$\vaf(\Omega,s) = \sup_{\straa \in \Straa^{FM}} \inf_{\strab \in \Strab^{FM}} \prob{s}^{\straa,\strab}(\Omega).$$


A strategy $\straa$ is \emph{optimal for finite memory} from a state $s$
if it uses finite memory and for all finite-memory strategies $\strab$ for player~$\mi$ it ensures that 
$\prob{s}^{\straa,\strab}(\Omega) \geq \vaf(\Omega,s)$.
We define analogously almost-sure winning with finite-memory strategies,
and the finite-memory value $\straavaf(\Omega,s)$ of $\straa$ in state~$s$
(against finite-memory strategies of player~$\mi$). 
We define the finite-memory value for player~$\mi$ by 
$\vbf(\Omega,s) = \inf_{\strab \in \Strab^{FM}} \sup_{\straa \in \Straa^{FM}} \prob{s}^{\straa,\strab}(\Omega)$
and the finite-memory value of strategy $\strab$ for player~$\mi$ by 
$\strabvbf(\Omega,s) = \sup_{\straa \in \Straa^{FM}} \prob{s}^{\straa,\strab}(\Omega)$.
We show in Theorem~\ref{theo:half-memoryless} for a large class of objectives 
(namely, prefix-independent shuffle-closed objectives) that 
the finite-memory value for player~$\ma$ and for player~$\mi$ coincide,
and allowing arbitrary strategies for player~$\mi$ (against finite-memory 
strategies for player~$\ma$) does not change the finite-memory value.

\smallskip\noindent{\bf Subgame-perfect strategies.}
Given a strategy $\straa$ for $\ma$, and a finite prefix $\rho = s_0 \cdots s_k$ 
of a play, we denote by $\straa_{\rho}$ the strategy that plays from the initial 
state $s_k$ what $\straa$ would play after the prefix $\rho$, i.e. such that
$\straa_{\rho}(s_k \cdot \rho') = \straa(\rho \cdot \rho')$ for all play prefixes $\rho'$, 
and $\straa_{\rho}(s \cdot \rho')$ is arbitrarily defined for all $s \neq s_k$.

A strategy $\straa$ for $\ma$ is \emph{subgame-perfect} if for all nonempty play 
prefixes $\rho \in S^+$, the strategy $\straa_{\rho}$ is optimal from the initial 
state $\Last(\rho)$.
Analogously, the strategy $\straa$ is \emph{subgame-perfect-for-finite-memory}
if all strategies $\straa_{\rho}$ are optimal-for-finite-memory strategies 
from $\Last(\rho)$.


\smallskip\noindent{\bf Value problems.}
Given an objective $\Omega$, a threshold $\lambda \in \rat$, and an initial state $s$, 
the \emph{value-strategy problem} 
asks whether there exists a strategy $\straa$ for player~$\ma$
such that $\straava(\Omega,s) \geq \lambda$ (or whether there exists a finite-memory 
strategy $\straa$ for player~$\ma$ such that $\straavaf(\Omega,s) \geq \lambda$).
The \emph{value problem} asks whether $\va(\Omega,s) \geq \lambda$ (resp., whether
$\vaf(\Omega,s) \geq \lambda$).

\smallskip\noindent{\bf End-component lemma.}
An important property of the end-components in MDPs is that for all 
strategies (with finite memory or not) with probability~1 the set
of states that are visited infinitely often along a play is
an end-component~\cite{CY95,deAlfaro97}. 
Given a play $\rho \in S^\omega$, let $\Inf(\rho)$ be the set of states 
that occur infinitely often in $\rho$.

\begin{lemma}\label{lem:end-component}
  \cite{CY95,deAlfaro97} Given an MDP $M$, for all states $s \in S$
  and all strategies $\straa \in \Straa$, we have
  $\prob{s}^{\straa}(\{ \rho \mid \Inf(\rho) \in \cale(M) \}) = 1$.
\end{lemma}

\begin{remark}[Key properties for MDPs]\label{rem:key}
The end-component lemma is useful in the analysis of MDPs with 
prefix-independent objectives, which can be decomposed into the analysis of the 
end-components (which have useful connectedness properties), 
and a reachability analysis to the end-components.
Moreover, suppose we consider prefix-independent objectives, and the 
MDP restricted to an end-component $U$. 
Then it follows from the results of~\cite{Cha07b} that either all states of $U$ 
have value~1 or all states of $U$ have value~0. 
Hence for prefix-independent objectives in MDPs, the optimal value is the optimal 
reachability probability to the {\em winning} end-components, where a
winning end-component is an end-component with value~1.
\end{remark}

\section{Half-Memoryless Result under Finite-Memory Strategies}\label{sec:half}
We show a general result that gives a sufficient condition for existence of memoryless
strategies (for one of the players) in games played with finite-memory strategies.

\smallskip\noindent{\em Comment on finite- vs. infinite-memory proof.}
The statement and proof structure of the result are similar to~\cite[Theorem~5.2]{GK14}
that established a sufficient condition for existence of memoryless optimal strategies 
in games played with arbitrary (infinite-memory) strategies. 
However, the proof uses different techniques.  
The key to establish the existence of memoryless strategies for one of the 
players is to first establish the existence of subgame-perfect strategies
for the other player.
We establish such a result in Lemma~\ref{lem:subgame-perfect} for finite-memory
strategies. 
Without the restriction of finite memory, only the existence of $\epsilon$-subgame-perfect 
strategies is known, and the proof requires intricate arguments and involved mathematical
machinery such as Doob's convergence theorem for martingales~\cite[Theorem~4.1]{GK14}.
Our proof is combinatorial and uses basic results on MDPs (e.g., discrete properties of 
end-components).

\smallskip\noindent{\em Key ideas of the proof.}
The proof of Lemma~\ref{lem:subgame-perfect} consists in constructing from a finite-memory
strategy~$\straa$ a strategy that is subgame-perfect-for-finite-memory by successively ``improving'' 
the value of the strategy $\straa_\rho$ for each finite prefix $\rho$. 
Improvements are obtained by modifying some transitions in the transducer defining $\straa$,
from the state reached after following the finite prefix $\rho$. The modification
of transitions does not change the memory space of the strategy, and since we consider 
finite-memory strategies, although there may be infinitely many finite prefixes $\rho$
where the strategy needs to be ``improved'', there is only a finite number of memory states
to consider for improvement, which guarantees the improvement process to terminate
and yields a subgame-perfect-for-finite-memory strategy.

\begin{figure}[!tb]
  \begin{center}
    \hrule  height .33pt

\begin{picture}(80,43)(0,2)

\gasset{Nw=6,Nh=6,Nmr=3,rdist=1, loopdiam=5}

\drawline[AHnb=0,arcradius=0, linegray=0](20,40)(5,10)
\drawline[AHnb=0,arcradius=0, linegray=0](20,40)(35,10)

\node[Nframe=n](n1)(13,36){$\GG_\straa$}
\node[Nframe=n](n1)(20,42){$\tuple{s_0,m_0}$}
\drawline[AHnb=0,arcradius=0, linegray=0.5](20,40)(19,35)(21,32)

\drawline[AHnb=0,arcradius=0, linegray=0.5](21,32)(18,28)(20,24)(18,20)
\node[Nframe=n](n1)(18,17){$\tuple{s,m_s}$}
\drawline[AHnb=0,arcradius=0, linegray=0.5](21,32)(23,28)(24,24)(27,20)(27,16)
\node[Nframe=n](n1)(27,13){$\tuple{s,m'}$}

\node[Nframe=n](n1)(40,5){$\winvalf{\mi_{\GG_{\straa}}}(\Omega, \tuple{s,m_s}) > \winvalf{\mi_{\GG_{\straa}}}(\Omega, \tuple{s,m'})$}

\node[Nframe=n](n1)(53,36){$\GG_{\straa'}$}
\node[Nframe=n](n1)(60,42){$\tuple{s_0,m_0}$}
\drawline[AHnb=0,arcradius=0, linegray=0](60,40)(45,10)
\drawline[AHnb=0,arcradius=0, linegray=0](60,40)(75,10)

\drawline[AHnb=0,arcradius=0, linegray=0.5](60,40)(59,35)(61,32)

\drawline[AHnb=0,arcradius=0, linegray=0.5](61,32)(58,28)(60,24)(58,20)
\node[Nframe=n](n1)(57,16){$\tuple{s,m_s}$}
\drawline[AHnb=0,arcradius=0, linegray=0.5](61,32)(63,28)(64,24)(67,20)(67,16)
\node[Nframe=n](n1)(67,13){$\tuple{s,m'}$}

\drawline[AHnb=1,arcradius=0, linegray=0.5, dash={1}0](67,16)(58,20)






\end{picture}
    \hrule  height .33pt
    \caption{Lemma~\ref{lem:subgame-perfect}: construction of a strategy $\straa'$ 
with higher value in subgames than the optimal-for-finite-memory strategy $\straa$. \label{fig:proof1}}
  \end{center}
\end{figure}
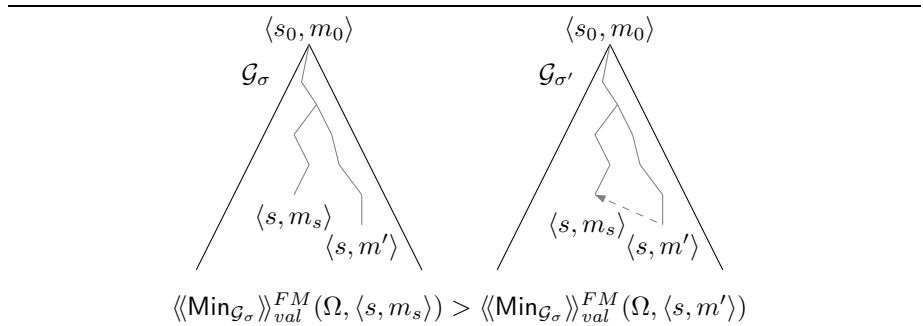

\begin{lemma} \label{lem:subgame-perfect}
In every stochastic game with a prefix-independent objective, there exists a 
subgame-perfect-for-finite-memory strategy for player~$\ma$.
\end{lemma}
\begin{proof}
Our proof is established using the following key steps:
\begin{compactenum}
\item Existence of an optimal-for-finite-memory strategy for player $\ma$.

\item Modification of the strategy for improvement of values after
finite prefixes.

\item The proof that the modification provides an 
improvement in two parts: 
once the strategy for player~$\ma$ is fixed, we have an MDP. 
In the MDP, we first show properties of the end-components, and 
second we provide bounds on the optimal reachability probability to 
the end-components to establish the improvement.

\end{compactenum}

\smallskip\noindent{\em Optimal-for-finite-memory strategy.}
We show the existence of a finite-memory strategy $\straa$ for player~$\ma$ 
in the game $\GG$ such that $\straa$ is optimal-for-finite-memory from every 
state for the prefix-independent objective~$\Omega$.
The fact that such a strategy always exists is as follows: 
it follows from~\cite[Theorem~4.3]{GH10} that it suffices to prove the result 
for almost-sure winning strategies.
Consider the set $Z$ of states with value~1 for finite-memory strategies. 
We need to show that there exists a finite-memory almost-sure winning strategy
in $Z$.
Let $0 < \epsilon < 1$, and consider a finite-memory strategy that ensures 
value at least $1-\epsilon$ from all states in $Z$. 
If a strategy can ensure positive winning from every state of a game, 
then it is almost-sure winning by the result of~\cite{Cha07b}.
The existence of an optimal-for-finite-memory strategy follows.

\smallskip\noindent{\em Notation.}
Consider an optimal-for-finite-memory strategy $\straa$.
Thus for all states $s$ of the game $\GG$ there
exists a memory value $m_s$ in the transducer of $\straa$ such that the value
of the objective~$\Omega$ in the MDP $\GG_{\straa}$ is the optimal finite-memory value, 
that is $\winvalf{\mi_{\GG_{\straa}}}(\Omega, \tuple{s,m_s}) = \winvalf{\ma_{\GG}}(\Omega, s)$
where the subscript in $\mi_{\GG_{\straa}}$ indicates that 
the value is computed in the MDP $\GG_{\straa}$ (which is a MDP for player~$\mi$)
while $\ma_{\GG}$ gives the optimal value for player~$\ma$ in the game~$\GG$.

\smallskip\noindent{\em Modification of the strategy.}
If the strategy $\straa$ is subgame-perfect-for-finite-memory, then the proof is done.
Otherwise, there exists a state $\tuple{s,m'}$ in $\GG_\straa$
with value below the optimal finite-memory value of $s$, 
namely such that $\winvalf{\mi_{\GG_{\straa}}}(\Omega, \tuple{s,m_s}) > \winvalf{\mi_{\GG_{\straa}}}(\Omega, \tuple{s,m'})$.
We construct an \emph{improved} strategy $\straa'$ as follows:
the strategy $\straa'$ plays like $\straa$ except
that when the state $\tuple{s,m'}$ is reached, the strategy $\straa'$ plays
like $\straa$ is playing from state $\tuple{s,m_s}$ (equivalently, we remove the outgoing transitions from state $\tuple{s,m'}$
in $\GG_\straa$, and replace them by a deterministic transition to state $\tuple{s,m_s}$ on all actions
to obtain $\GG_{\straa'}$, as illustrated in \figurename~\ref{fig:proof1}). 
Note that the new strategy $\straa'$ has the same memory set as $\straa$. We show below that
the value of every state in $\GG_{\straa'}$ is at least as large as the value of the 
same state in $\GG_\straa$ $(\star)$. It follows that the value of state $\tuple{s,m'}$ in $\GG_{\straa'}$
is the optimal finite-memory value from $s$, and by repeating the same construction
in every state where the value is below the optimal finite-memory value, we obtain
(in finitely many steps) a subgame-perfect-for-finite-memory strategy for player~$\ma$.

\smallskip\noindent{\em Proof of $(\star)$.}
We proceed with the proof of $(\star)$, which has two steps as mentioned above.
We first define the notion of value class. 

\smallskip\noindent{\em Value class and properties.}
In the MDP $\GG_\straa$, a \emph{value class} is a maximal subset of states that have 
the same value (defined as the infimum over the strategies of player~$\mi$).
The following property holds in $\GG_\straa$, for every state $l = \tuple{\cdot,\cdot}$, and action $a \in A$:
consider the value class of $l$, if there is an $a$-successor of $l$ 
in a lower value class, then there is also an $a$-successor of $l$ in a 
higher value class (\figurename~\ref{fig:proof2}). 
If we consider the partition defined by the value classes
in $\GG_\straa$, this property also holds in the modified MDP $\GG_{\straa'}$ 
corresponding to strategy $\straa'$, because the new deterministic transition 
(dashed edge of  \figurename~\ref{fig:proof1}) goes to a higher value class.

\smallskip\noindent{\em Properties of end-components.}
Now, we claim that in the modified MDP $\GG_{\straa'}$ every end-component
is included in some value class (of the original MDP $\GG_\straa$).
We show this by contradiction (see also \figurename~\ref{fig:proof2}). 
Assume that there is an end-component $C$ in $\GG_{\straa'}$
with non-empty intersection with different value classes (of the original MDP $\GG_\straa$). 
Let $x \in C$ be a state of $C$ with largest value. Since $C$ is strongly connected,
there is a path from $x$ to a lower value class, and on this path there is a state $y \in C$ 
with largest value that has an $a$-successor $z$ with lower value (for some $a \in A_C(y)$). 
It follows that $y$ has also an $a$-successor with higher value, according to the above property. This
successor is outside $C$ since there is no larger value class in $C$ than the value class of $y$. 
This is in contradiction with the fact that end-components are closed sets (and that $a \in A_C(y)$). 
We conclude that in $\GG_{\straa'}$ every end-component
is included in some value class (of the original MDP $\GG_\straa$).
Therefore, the value of each end-component in $\GG_{\straa'}$ is at least as large
as the value of the value class containing it (in $\GG_\straa$). 
It also follows that the new deterministic transitions from $\tuple{s,m'}$ to 
$\tuple{s,m_s}$ do not belong to any end-component in $\GG_{\straa'}$.

\begin{figure}[!tb]
  \begin{center}
    \hrule  height .33pt
    \renewcommand{\sb}[1]{\scalebox{0.75}[1]{#1}}

\begin{picture}(80,46)(0,2)

\gasset{Nw=6,Nh=6,Nmr=3,rdist=1, loopdiam=5}

\node[Nmarks=n,Nw=15,Nh=30,Nmr=0](n1)(10,26){}
\node[Nframe=n,Nw=20,Nh=10,Nmr=0](nl)(10,6){\begin{tabular}{c}lowest\\value class\end{tabular} }
\node[Nframe=n,Nw=0,Nh=0,Nmr=0](nC)(-1,40){$\GG_{\straa}$}
\node[Nmarks=n,Nw=15,Nh=30,Nmr=0](n1)(30,26){}
\node[Nmarks=n,Nw=15,Nh=30,Nmr=0](n1)(50,26){}
\node[Nmarks=n,Nw=15,Nh=30,Nmr=0](n1)(70,26){}
\node[Nframe=n,Nw=20,Nh=10,Nmr=0](nh)(70,6){\begin{tabular}{c}highest\\value class\end{tabular} }
\drawedge[ELpos=50, ELside=l, ELdist=.5, curvedepth=0, linegray=0.5, dash={1}0](nl,nh){increasing}
\drawedge[ELpos=50, ELside=r, ELdist=.5, curvedepth=0, linegray=0.5, dash={1}0](nl,nh){value class}

\node[Nframe=n,Nmarks=n,Nw=0,Nh=0,Nmr=0, ExtNL=y, NLangle=60, NLdist=.5](n1)(54,20){$l$}
\rpnode[Nmarks=n](np)(50,14)(4,1){}
\node[Nframe=n,Nw=3,Nh=3,Nmr=1.5](n2)(70,15){}
\node[Nframe=n,Nw=3,Nh=3,Nmr=1.5](n3)(30,15){}

\drawedge[ELpos=50, ELside=l, curvedepth=0](n1,np){$a$}
\drawedge[ELpos=50, ELside=l, curvedepth=-5](np,n2){}
\drawedge[ELpos=50, ELside=l, curvedepth=5](np,n3){}

\node[Nmarks=n,Nw=43,Nh=14,Nmr=0, linegray=0.5, ExtNL=y, NLangle=162, NLdist=1](nC)(35,38){$C$}

\rpnode[Nmarks=n, ExtNL=y, NLangle=270, NLdist=1](ny)(46,39)(4,1){}
\node[Nframe=n,Nw=0,Nh=0,Nmr=0](nL)(31,37){}
\node[Nframe=n,Nw=3,Nh=3,Nmr=1.5](nR)(60.5,43.5){$\times$}

\node[Nframe=n,Nw=0,Nh=0,Nmr=1.5, ExtNL=y, NLangle=270, NLdist=1](ndummy)(48,35.5){$y$}
\drawedge[AHnb=0, ELpos=50, ELside=l, curvedepth=0, linegray=0.5](ny,ndummy){}
\drawline[AHnb=0, arcradius=0, linegray=0.5](48,35.5)(50,38)(52,36)(55,38)
\node[Nframe=n,Nw=0,Nh=0,Nmr=1.5](nx)(55,39.5){$x$}
\drawline[AHnb=0, arcradius=0, linegray=0.5](31,37)(28,38)(26,35)(24,36)
\node[Nframe=n,Nw=0,Nh=0,Nmr=1.5](nx)(31,35.5){$z$}

\drawedge[ELpos=50, ELside=l, curvedepth=-4](ny,nL){}
\drawedge[ELpos=50, ELside=l, curvedepth=2.5](ny,nR){}





\end{picture}
    \hrule  height .33pt
    \caption{Lemma~\ref{lem:subgame-perfect}: value-class analysis. No end-component $C$ can lie across several value classes.\label{fig:proof2}}
  \end{center}
\end{figure}

\smallskip\noindent{\em Optimal reachability probability.} 
The key steps to obtain the bound on optimal reachability probability is as follows:
we observe that the optimal reachability probability in MDPs is characterized by a minimizing 
linear-programming solution, and we show that the solution before the modification 
is a feasible solution after the modification.
We now present the details.

\smallskip\noindent{\em Optimal value via optimal reachability.}
We show that the value of the state $\tuple{s,m'}$ in $\GG_{\straa'}$ is 
strictly greater than the value of $\tuple{s,m}$ in $\GG_{\straa}$ (for player~$\ma$).
Let $S_{losing}$ be the union of all end-components in $\GG_\straa$ with value $0$
for the prefix-independent objective~$\Omega$ (thus losing for player~$\ma$, and winning for player~$\mi$).
By Remark~\ref{rem:key}, the optimal value for player~$\mi$ in the MDP
is the optimal reachability probability to $S_{losing}$.

\smallskip\noindent{\em Optimal reachability probability to $S_{losing}$.}
Consider the following linear program in $\GG_\straa = \tuple{S', (\emptyset, S'_{\mi}), A, \PP'}$ 
that computes the value (for player~$\mi$) of each state $l \in S'$ of $\GG_{\straa}$ in variable $x_l$,
by solving a reachability problem to the states in $S_{losing}$: 
\begin{itemize}
\item[] minimize $\sum_{l \in S'} x_l$
\item[] $x_l \geq \sum_{k \in S'} \PP'(l,a)(k)\cdot x_k$ for all $l \in S', a \in A$
\item[] $x_l = 1$ for all $l \in S_{losing}$
\end{itemize}
The correctness of the linear program to compute optimal reachability probability 
is standard~\cite{FV97}.
Let $x^*$ be an optimal solution of this linear program. Note that the values
are computed for player~$\mi$, and thus $x^*_l = 1 - \vaf(\Omega, l)$.
It follows that $x^*_{\tuple{s,m_s}} < x^*_{\tuple{s,m'}}$.

\smallskip\noindent{\em Feasible solution.}
Consider the modified MDP $\GG_{\straa'}$ (with same state space as $\GG_\straa$), 
in which the union of end-components with value $0$ is contained in $S_{losing}$. 
Therefore, considering the same linear program for $\GG_{\straa'}$ provides 
an upper bound on the new value (for player~$\mi$).
For each $l\in S'$, define $y_l = \left\{ \begin{array}{ll} x^*_l & \text { if }l \neq \tuple{s,m'} \\  
x^*_{\tuple{s,m_s}} & \text { if }l = \tuple{s,m'}  \end{array}  \right.$

\noindent Then $(y_l)_{l \in S'}$ is a feasible solution to the linear program for $\GG_{\straa'}$,
and for the optimal solution $y^*$, we have $y^*_l \leq y_l \leq x^*_l$ (and
for $l' = \tuple{s,m'}$ we have $y^*_{l'} \leq y_{l'} < x^*_{l'}$). Since 
$y^*_{l'}$ is only an upper bound of the new value of $s$ for player~$\mi$ in $\GG_{\straa'}$,
it shows that the value improved for player~$\ma$ in every state.  
Since the value of $\tuple{s,m_s}$ in $\GG_{\straa}$ was the optimal finite-memory value,
it follows that in $\GG_{\straa'}$ the value of $\tuple{s,m_s}$ is also the optimal finite-memory value.
Since all transitions of $\tuple{s,m'}$ lead to $\tuple{s,m_s}$, the value
of $\tuple{s,m'}$ in $\GG_{\straa'}$ is the optimal finite-memory value from $s$,
which concludes the proof of~$(\star)$.
\end{proof}

The result of~\cite[Theorem~5.2]{GK14} shows that in games where the players 
are allowed to use arbitrary strategies (thus not restricted to finite-memory 
strategies), memoryless optimal strategies exist for player~$\mi$ if the objective 
of player~$\ma$ is prefix-independent and closed under shuffling. 
The proof of this result uses an analogue of Lemma~\ref{lem:subgame-perfect}
for arbitrary strategies, and relies on edge induction, a technique that
became standard~\cite{CD12,GK14,GZ05,Kop06}. 
The shape of the argument is not specific to games with arbitrary strategies:
in games where the players are restricted to finite-memory strategies, 
we can follow the same line of proof (using Lemma~\ref{lem:subgame-perfect}) 
to show that if the objective of a player is prefix-independent and 
closed under shuffling, then memoryless optimal strategies exist for the other player.

\begin{theorem}\label{theo:half-memoryless}
In stochastic games, if the objective $\Omega$ of player~$\ma$ is prefix-independent and
closed under shuffling, and player~$\ma$ is restricted to finite-memory strategies,
then player~$\mi$ has a memoryless optimal-for-finite-memory strategy (as well as a 
memoryless optimal strategy), and determinacy holds under finite-memory strategies.
More precisely, for all states $s$ we have:
$$\vaf(\Omega,s) = \vbf(\Omega,s) =: v(s), \text{and}$$
$$ \sup_{\straa \in \Straa^{FM}} \inf_{\strab \in \Strab} \prob{s}^{\straa,\strab}(\Omega,s) = v(s) = 
\inf_{\strab \in \Strab^{PM}} \sup_{\straa \in \Straa^{FM}} \prob{s}^{\straa,\strab}(\Omega,s).
$$
%
\end{theorem}

\smallskip\noindent{\em Significance of Theorem~\ref{theo:half-memoryless}.}
We first remark on the significance of the result, and then present the main steps of the 
proof.
First, the result establishes determinacy for finite-memory strategies 
i.e., $\vaf(\Omega,s) = \vbf(\Omega,s)=v(s)$, which implies that even for finite-memory strategies
the order of sup and inf can be exchanged. 
However, note that the finite-memory value is different from the value under infinite-memory 
strategies, and the determinacy for finite-memory does not follow from the determinacy
for infinite-memory strategies.
Second, $\sup_{\straa \in \Straa^{FM}} \inf_{\strab \in \Strab} \prob{s}^{\straa,\strab}(\Omega,s) = v(s)$
implies that as long as player~$\ma$ is restricted to finite-memory strategies, whether 
player~$\mi$ uses finite-memory or infinite-memory strategies does not matter.
Finally, $v(s)= \inf_{\strab \in \Strab^{PM}} \sup_{\straa \in \Straa^{FM}} \prob{s}^{\straa,\strab}(\Omega,s)$
implies that against finite-memory strategies of player~$\ma$ there exists a pure memoryless
strategy for player~$\mi$ that is optimal (even considering all infinite-memory strategies 
for player~$\mi$).


\smallskip\noindent{\em Main steps of the proof.}
We present the key steps of the proof of Theorem~\ref{theo:half-memoryless}, 
and we show that the argument in the proof of~\cite[Theorem~5.2]{GK14} (which we refer to 
for the precise technical steps) can be adapted for finite-memory strategies.
The key steps are: 
(i)~induction on the number of player-$\mi$ states;
(ii)~creating different games for different choices at a player-$\mi$ state, in which
player~$\mi$ has memoryless optimal strategies by induction hypothesis;
and (iii)~showing the value of the original game is at least the minimum of the 
value of the different games, thus memoryless strategies suffice.

\smallskip\noindent{\em Induction on player-$\mi$ states.}
The proof is by induction on the number of states of player~$\mi$. The base case
$\abs{S_{\mi}} = 0$ corresponds to games with only states of player~$\ma$. 
The result holds trivially in that case (the empty strategy of player~$\mi$ 
is memoryless). For the induction step, assume that the result holds for all games
with $\abs{S_{\mi}} < k$, and consider a game $\GG$ with $\abs{S_{\mi}} = k$. 

\smallskip\noindent{\em Different games for different choices.}
We explain the rest of the proof assuming the action set contains only two actions, 
that is $A = \{a,b\}$. The proof is the same for an arbitrary finite set of actions, 
with more complication in the notation. 
In $\GG$, consider a state $\hat{s} \in S_{\mi}$ of player~$\mi$ and construct
two games $\GG_a$ and $\GG_b$ obtained from $\GG$ by removing $\hat{s}$ and by
replacing the incoming transitions to $\hat{s}$ by transitions to its $a$-successors
and $b$-successors respectively. The transition function of $\GG_x$ (for $x \in \{a,b\}$)
is defined by $\PP_x(s,c)(s') = \PP(s,c)(s') + \PP(s,c)(\hat{s}) \cdot \PP(\hat{s},x)(s')$
for all $s,s' \in S \setminus \{\hat{s}\}$, and all actions $c \in A$.

\smallskip\noindent{\em Value of original game at least the minimum of the value 
of the two games.}
In $\GG_a$ and $\GG_b$ the number of states of player~$\mi$ is $k-1$. 
Hence by the induction hypothesis there exist memoryless strategies $\strab^{\GG_a}$ and $\strab^{\GG_b}$
for player~$\mi$ that are optimal-for-finite-memory (as well as optimal among the infinite-memory strategies) 
in $\GG_a$ and $\GG_b$ respectively.
The proof proceeds by showing that in the game $\GG$, player~$\mi$ cannot obtain
a lower (i.e., better) value than in one of the games $\GG_a$ or $\GG_b$,
that is for all strategies $\strab$ of player~$\mi$, for all states $s \neq \hat{s}$
we have\footnote{We assume that the value $\mystravaf{\strab^{\GG}}(\Omega,s)$ of 
a strategy $\strab^{\GG}$ is computed in the game $\GG$ in superscript.}:
\begin{equation}
\mystravaf{\strab^{\GG}}(\Omega,s)\!\geq\!\min \big\{\mystravaf{\strab^{\GG_a}}(\Omega,s),\mystravaf{\strab^{\GG_b}}(\Omega,s)\big\}.
\label{eq:value}
\end{equation}
To show this, we consider subgame-perfect-for-finite-memory strategies $\straa_a$ and $\straa_b$ 
for player~$\ma$ in games $\GG_a$ and $\GG_b$ respectively (which exist by Lemma~\ref{lem:subgame-perfect}), 
and we construct
a finite-memory strategy $\straa$ in $\GG$ that achieves, against all strategies 
$\strab$, a value at least as large as either $\straa_a$ in $\GG_a$ or $\straa_b$ 
in $\GG_b$. Intuitively, $\straa$ switches between $\straa_a$ and $\straa_b$, 
playing according to $\straa_a$ when in the last visit to $\hat{s}$ player~$\mi$ 
played action $a$ (thus as in $\GG_a$),
and playing according to $\straa_b$ when in the last visit to $\hat{s}$ player~$\mi$ 
played action $b$ (thus as in $\GG_b$). 
To formally define $\straa$, given a play prefix in $\GG$ we use projections 
onto plays in $\GG_a$ (resp., $\GG_b$) that erase all sub-plays between
successive visits to $\hat{s}$ where action~$b$ (resp., action~$a$) was played in $\hat{s}$.
Note that $\straa$ uses finite memory.
The plays compatible with $\straa$ and $\strab$ are shuffles of plays compatible 
with $\straa_a$ in $\GG_a$ and plays compatible with $\straa_b$ in $\GG_b$,
and since the objective $\Omega$ is closed under shuffling, the probability
measure of the plays satisfying the objective in $\GG$ is no lower than
the value of either games $\GG_a$ or $\GG_b$:
$$\prob{s}^{\straa,\strab}(\Omega) \geq \min \big\{\mystravaf{\strab^{\GG_a}}(\Omega,s),\mystravaf{\strab^{\GG_b}}(\Omega,s)\big\}.$$
It follows that \eqref{eq:value} holds, and thus the optimal-for-finite-memory (as well as optimal among 
infinite-memory strategies) strategies 
in the games $\GG_a$ and $\GG_b$ (extended to play $a$ and $b$ respectively in $\hat{s}$) are sufficient
for player~$\mi$ in $\GG$. Therefore by the induction hypothesis, memoryless strategies
are sufficient for player~$\mi$ to achieve the optimal finite-memory value, let 
$\strab$ be such a strategy.
By the same argument and using the induction hypothesis, for the
finite-memory strategy $\straa$ for player~$\ma$ in $\GG$ we have 
$\straava(\Omega,s)=\straavaf(\Omega,s)=\strabvbf(\Omega,s)$, 
which gives $\vaf(\Omega,s) = \vbf(\Omega,s)$.
Note that our proof handled that the strategies for player~$\mi$ are allowed to be 
infinite-memory, and the result still holds.

\begin{figure}[!tb]
  \begin{center}
    \hrule  height .33pt
\begin{picture}(30,12)(0,1)

\gasset{Nw=6,Nh=6,Nmr=3,rdist=1, loopdiam=5}

\node[Nmarks=n](n1)(5,8){$s_1$}
\node[Nmarks=n, ExtNL=y, NLangle=270, NLdist=1](n1)(5,8){$1$}
\node[Nmarks=r, Nmr=0](n2)(25,8){$s_2$}
\node[Nmarks=r, Nmr=0, ExtNL=y, NLangle=270, NLdist=1](n2)(25,8){$-1$}

\drawloop[ELside=l, loopCW=y, loopangle=180](n1){}
\drawedge[ELpos=50, ELside=l, curvedepth=4](n1,n2){}

\drawedge[ELpos=50, ELside=l, curvedepth=4](n2,n1){}
\drawloop[ELside=l, loopCW=y, loopangle=0](n2){}





\end{picture}
    \hrule  height .33pt
    \caption{A game with prefix-independent objective $\Buchi(s_2) \land (\coBuchi(s_2) \lor \MeanSup)$ that is not determined under finite-memory strategies.\label{fig:not-determined}}
  \end{center}
\end{figure}

\begin{remark}\label{rem:nodet}
The determinacy result of Theorem~\ref{theo:half-memoryless},
which allows to switch the $\sup$ and $\inf$ operators ranging over finite-memory
strategies, is true for prefix-independent shuffle-closed objectives.
We present an example to show that such a result does not hold  
for general prefix-independent objectives that are not closed under shuffling. 
Consider the game of \figurename~\ref{fig:not-determined}, 
with the objective  $\Omega = \Buchi(s_2) \land (\coBuchi(s_2) \lor \MeanSup)$
where $\Buchi(s_2)$ is the set of plays that visit $s_2$ infinitely often, 
and $\coBuchi(s_2)$ is the set of plays that eventually stay in $s_2$ forever. 
Note that the game is even non-stochastic.
We show that $\vaf(\Omega,s_1) = 0$ and $\vbf(\Omega,s_1) = 1$. Intuitively, 
after either player fixed a finite-memory strategy, the other player can win
using slightly more memory than the first player (but still finite memory).
For all finite-memory strategies $\straa$ of player~$\ma$, either $(i)$ there exists a compatible
play that eventually stays forever in $s_1$, and then the objective $\Buchi(s_2)$
is violated, or $(ii)$ $s_2$ is visited infinitely often in all compatible plays 
and player~$\mi$ can ensure with a finite-memory strategy that both objectives $\MeanSup$ and 
$\coBuchi(s_2)$ are violated by staying in $s_2$ one more time than player~$\ma$ 
stayed in $s_1$, and then going back to $s_1$. It follows that $\straavaf(\Omega,s_1) = 0$.
Analogously, against all finite-memory strategies $\strab$ of player~$\mi$,
player~$\ma$ can ensure that the objective $\Omega$ is satisfied (by staying
in $s_1$ one more time than player~$\mi$ stayed in $s_2$, and then going to $s_2$), thus 
$\strabvbf(\Omega,s_1) = 1$. Hence $\vaf(\Omega,s_1) \neq \vbf(\Omega,s_1)$ and the
game of \figurename~\ref{fig:not-determined} is not determined under finite-memory strategies.
\end{remark}

\smallskip\noindent{\em Upper bound on memory.} 
We now show that for prefix-independent shuffle-closed objectives, 
the memory required for player~$\ma$ is exponential as compared
to the memory required for the same objective in MDPs.
If there are $k$ states for player~$\mi$, then the optimal-for-finite-memory 
strategy $\straa$ constructed for player~$\ma$ in the proof of Theorem~\ref{theo:half-memoryless}
is as follows: it considers strategies in the choice-fixed games 
($\GG_a$ and $\GG_b$) with $k-1$ states for player~$\mi$, and the strategy in 
the original game considers projections of plays and then copies the strategies 
of the choice-fixed games. 
Thus the memory required for player~$\ma$ in games with $k$ states for player~$\mi$ is the union
of the memory required for the choice-fixed games with $k-1$ states, and there 
are at most $\abs{A}$ such choice-fixed games.
If we denote by $M(k)$ the memory required for player~$\ma$ in games with $k$ 
player-$\mi$ states, then the following recurrence is satisfied:
\[
M(k) = \abs{A} \cdot M(k-1). 
\]
Note that $M(0)$ represents the memory bound for MDPs, and thus we get a bound on 
$M(k)=\abs{A}^k \cdot M(0)$ in games that is greater than the 
memory bound for MDPs by an exponential factor.

\begin{theorem}\label{theo:membou}
In stochastic games with a prefix-independent shuffle-closed objective $\Omega$, 
an upper bound on the memory required for optimal-for-finite-memory strategies
is $\abs{A}^{\abs{S_{\mi}}} \cdot M(0)$, where $M(0)$ is an upper bound on memory
required for objective $\Omega$ in MDPs.
\end{theorem}

\section{Generalized Mean-Payoff Objectives under Finite-Memory Strategies}\label{sec:finmem}

In generalized-mean-payoff games, infinite-memory strategies are more powerful 
than finite-memory strategies, even in 1-player games with only deterministic 
transitions, i.e., graphs~\cite[Lemma~7]{VCDHRR15}.\footnote{In the example of~\cite[Lemma~7]{VCDHRR15}
all finite-memory strategies have winning probability~0 while there exists an almost-sure 
winning strategy (with infinite memory).}
It follows that in general $\va(\Omega, s) \neq \vaf(\Omega, s)$ in generalized-mean-payoff games
(for both $\Omega = \MeanSup$ and $\Omega = \MeanInf$).
In this section, we consider the value problem for finite-memory
strategies, and present complexity results showing that the problem is in PTIME for
MDPs, and is coNP-complete for games.
Finally we present optimal bounds for memory required in 2\half-player games.

\subsection{Generalized mean-payoff objectives under finite-memory in MDPs}

We consider the value problem for finite-memory strategies
in MDPs with generalized mean-payoff objectives. 
First we show that randomized memoryless strategies are as powerful as 
finite-memory strategies, and then using this result we show that the 
value problem can be solved in polynomial time.

Note that in finite-state Markov chains with a fixed reward function, from all states $s$, 
the probability that the conjunction $\MeanSup$ of mean-payoff-sup objectives holds from $s$ is 
the same as the probability that the conjunction $\MeanInf$ of mean-payoff-inf objectives holds 
from $s$~\cite{FV97}. It follows that in MDPs with finite-memory strategies, the 
value for mean-payoff-sup and mean-payoff-inf objectives coincides, thus
$\vaf(\MeanSup, s) = \vaf(\MeanInf, s)$ for all states $s$.

\smallskip\noindent{\em Key ideas.}
Let $M = \tuple{S, A, \PP}$ be an MDP and $\rwd: S \to \real^k$
be a reward function. The key ideas to show that randomized memoryless strategies
are sufficient for generalized mean-payoff objectives are: 
(i)~first observe that the mean-payoff value of a play depends only on the frequency of 
occurrence of each state, (ii)~under finite-memory strategies the frequencies are well defined 
(with probability~1) for each state and action, and (iii)~given the frequencies of a
finite-memory strategy, a randomized memoryless strategy that plays at every state an 
action with probability proportional to the given frequencies achieves the same frequencies
as the finite-memory strategy.

\smallskip\noindent{\em Definition of frequencies.}
In this paragraph, we assume that plays are sequences of alternating
states and actions, and the probability measure $\prob{s}^{\straa}(\cdot)$
is (uniquely) defined over Borel sets of such sequences.
Given a play $\rho = s_0 a_0 s_1 a_1 \dots \in (S A)^\omega$, let $N_i^{\rho}(s)$ be the 
number of occurrences of state $s$ in $\rho$ up to position $i$, and
let $N_i^{\rho}(s,a)$ be the number of occurrences of the pair $(s,a)$ of state $s$
followed by action $a$ in $\rho$ up to position~$i$.
The frequency $f_s^\rho$ of a state $s$ in $\rho$ and 
the frequency $f_{s,a}^\rho$ of a pair $(s,a)$ are:
$$f_s^\rho = \lim_{i \to \infty} \frac{N_i^{\rho}(s)}{i} \text{ (if the limit exists),}$$ 
$$f_{s,a}^\rho = \lim_{i \to \infty} \frac{N_i^{\rho}(s,a)}{i} \text{ (if the limit exists)}.$$

\smallskip\noindent{\em Mean-payoff values from frequencies.}
It is easy to see that if the frequency $f_s^\rho$ is defined, then 
the mean-payoff values $v_j = \limsup_{n \to \infty} \frac{1}{n} \cdot \sum_{i = 0}^{n-1}\rwd_j(s_{i})$
of $\rho$ (for $j=1,\ldots,k$) can be computed from the values of $f_s^\rho$
as follows: $v_j = \sum_{s \in S} f_s^\rho \cdot \rwd_j(s)$.

\smallskip\noindent{\em Almost-sure existence of frequencies.}
We note that in finite-state Markov chains (and in finite-state MDPs with finite-memory
strategies) the frequencies are defined with probability~1~\cite{FV97}:
$$\prob{s}^{\straa}(\{\rho \mid f_s^\rho \text{ and } f_{s,a}^\rho \text{ are defined for all } s \in S \text{ and } a \in A \}) = 1.$$
Given a finite-memory strategy $\straa$ in an MDP $M$, consider a closed recurrent set $B$
of the Markov chain $M_{\straa}$ obtained by playing $\straa$ in 
$M$\footnote{The state space of $M_{\straa}$ is the Cartesian product of the set
of states of $M$ and of the memory set of $\straa$. However, 
we say that $s \in B$ if $(s,m)$ is a state of $B$ for some memory value $m$ of $\straa$.}. 
Then, given $s_0 \in B$ the frequencies are fixed in $B$, that is there exist (unique) frequency 
vectors $\lambda: B \to [0,1]$ and $\mu: B \times A \to [0,1]$ such that:
$$
\begin{array}{l}
\prob{s_0}^{\straa}(\{\rho \mid \frac{f_{s,a}^\rho}{f_s^\rho} = \mu(s,a) \text{ for all } s \in B \text{ and } a \in A \}) = 1, \text{ and } \\[+3pt]
\prob{s_0}^{\straa}(\{\rho \mid f_s^\rho = \lambda(s) \text{ for all } s \in B \}) = 1.
\end{array}
$$

\begin{figure}[!tbp]
  \centering
  \hrule height .33pt
\begin{picture}(120,34)(0,-1)

\gasset{Nw=10,Nh=6,Nmr=3,rdist=1, loopdiam=5}

\node[Nmarks=n, Nw=8](n1)(25,16){\sb{$(0,0)$}}
\node[Nmarks=n](n2)(5,16){\sb{$(-3,4)$}}
\rpnode[Nmarks=n](np)(40,16)(4,2){}
\node[Nmarks=n](n3)(55,24){\sb{$(-2,1)$}}
\node[Nmarks=n](n4)(55,8){\sb{$(3,-2)$}}

\nodelabel[ExtNL=y, NLangle=270, NLdist=1](n1){$s_1$}
\nodelabel[ExtNL=y, NLangle=270, NLdist=1](n2){$s_2$}
\nodelabel[ExtNL=y, NLangle=270, NLdist=1](n3){$s_3$}
\nodelabel[ExtNL=y, NLangle=270, NLdist=1](n4){$s_4$}


\drawedge[ELpos=50, ELside=r, curvedepth=-5](n1,n2){$f_1$}
\drawedge[ELpos=50, ELside=r, curvedepth=-5](n2,n1){$f_2$}

\drawedge[ELpos=66, ELside=r, curvedepth=0](n1,np){$f_3$}
\drawedge[ELpos=30, ELside=l, ELdist=.5, curvedepth=0](np,n3){{\scriptsize $\sfrac{1}{2}$}}
\drawedge[ELpos=30, ELside=r, ELdist=.5, curvedepth=0](np,n4){{\scriptsize $\sfrac{1}{2}$}}

\drawedge[ELpos=50, ELside=r, curvedepth=-5](n3,n1){$f_4$}
\drawedge[ELpos=50, ELside=l, curvedepth=5](n4,n1){$f_5$}


\put(70,24){\makebox(0,0)[l]{$(E\ref{n1})$}}
\put(77,24){\makebox(0,0)[l]{$\left\{\begin{array}{l}f_1 = f_2 \\
f_1 + f_3 = f_2 + f_4 + f_5 \\
f_4 = \frac{f_3}{2}\\
f_5 = \frac{f_3}{2}\\
\end{array}\right.$}}

\put(70,10){\makebox(0,0)[l]{$(E\ref{n2})$}}
\put(77,10){\makebox(0,0)[l]{\hspace{1pt}$\left\{\begin{array}{l}-3 f_2 - 2 f_4 + 3 f_5 \geq 0 \\
4 f_2 + f_4 - 2 f_5 \geq 0 \\
\end{array}\right.$}}

\put(70,2){\makebox(0,0)[l]{$(E\ref{n3})$}}
\put(77,2){\makebox(0,0)[l]{$\phantom{\{}f_1 + f_2 + f_3 + f_4 + f_5 = 1$}}





\end{picture}
  \hrule height .33pt
    \caption{Linear program for an MDP with two-dimensional mean-payoff objective (the 
	constraints $f_i \geq 0$ for $i=1,\dots,5$ are omitted in the figure).\label{fig:frequency}}
\end{figure}
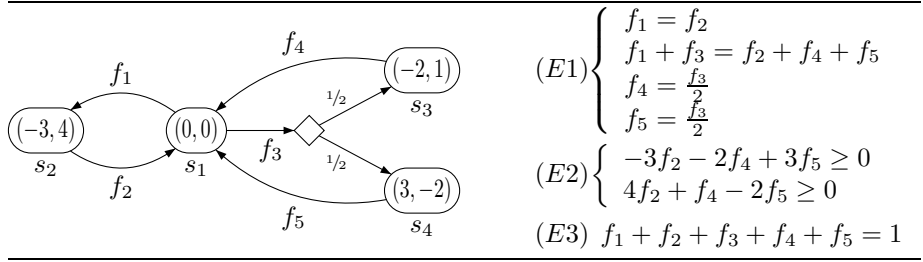

\begin{figure}[!tbp]
  \centering
  \hrule height .33pt  

\begin{picture}(50,23)(0,2)

\gasset{Nw=10,Nh=6,Nmr=3,rdist=1, loopdiam=5}

\node[Nmarks=n](n1)(5,12){\sb{$(-1,1)$}}
\node[Nmarks=n, Nw=11](n2)(25,12){\sb{$(-1,-1)$}}
\node[Nmarks=n](n3)(45,12){\sb{$(1,-1)$}}

\nodelabel[ExtNL=y, NLangle=270, NLdist=1](n1){$s_1$}
\nodelabel[ExtNL=y, NLangle=270, NLdist=1](n2){$s_2$}
\nodelabel[ExtNL=y, NLangle=270, NLdist=1](n3){$s_3$}

\drawloop[ELside=l, loopCW=y](n1){$f_1$}
\drawloop[ELside=l, loopCW=y](n3){$f_6$}

\drawedge[ELpos=50, ELside=r, curvedepth=-5](n1,n2){$f_3$}
\drawedge[ELpos=50, ELside=r, curvedepth=-5](n2,n1){$f_2$}
\drawedge[ELpos=50, ELside=l, curvedepth=5](n3,n2){$f_5$}
\drawedge[ELpos=50, ELside=l, curvedepth=5](n2,n3){$f_4$}






\end{picture}
  \hrule height .33pt
   \caption{The (disjoint) union of two end-components corresponds to a solution 
of $\LP$ ($f_1 = f_6 = \frac{1}{2}$ and $f_2 = f_3 = f_4 = f_5 = 0$). However, no single end-component is a solution.\label{fig:frequency-not-connected}} 
\end{figure}
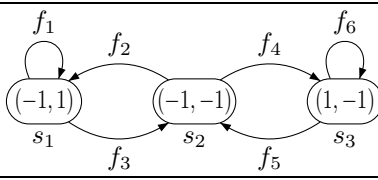

\smallskip\noindent{\em Strategies given frequencies.}
In the closed recurrent set $B$, consider a finite-memory strategy $\straa'$ that plays actions
with the same frequencies as $\straa$, that is such that:
$$
\prob{s_0}^{\straa'}(\{\rho \mid \frac{f_{s,a}^\rho}{f_s^\rho} = \mu(s,a) \text{ for all } s \in B \text{ and } a \in A \}) = 1
$$
Then the frequencies of the states in $B$ are given by $\lambda$:
$$
\prob{s_0}^{\straa'}(\{\rho \mid f_s^\rho = \lambda(s) \text{ for all } s \in B \}) = 1.
$$

\smallskip\noindent{\em Sufficiency of randomized memoryless strategies.}
Now, consider a randomized memoryless strategy $\bar{\straa}$ such that
$\bar{\straa}(s) = d_s$ for all $s \in B$ where $d_s \in \DD(A)$ is such that 
$d_s(a) = \mu(s,a)$ for all $a \in A$. Thus $\bar{\straa}$ plays each action 
with the frequencies given by $\mu$, as in $\straa$. 
It follows that 
$$
\prob{s_0}^{\bar{\straa}}(\{\rho \mid \frac{f_{s,a}^\rho}{f_s^\rho} = \mu(s,a) \text{ for all } s \in B \text{ and } a \in A \}) = 1,
$$
and thus
$$
\prob{s_0}^{\bar{\straa}}(\{\rho \mid f_s^\rho = \lambda(s) \text{ for all } s \in B \}) = 1.
$$
Since the mean-payoff value depends only on the frequency of each state,
it follows that $\prob{s_0}^{\straa}(\MeanSup) = \prob{s_0}^{\bar{\straa}}(\MeanSup)$
for all $s_0 \in B$, and $\bar{\straa}$ is a randomized memoryless strategy.
Thus randomized memoryless strategies can achieve the same values as arbitrary
finite-memory strategies.
Note that the result holds for conjunctions of mean-payoff-inf objectives 
as well (by the same proof).
By Remark~\ref{rem:key}
the winning probability from an initial state is the maximum probability to
reach end-components with value~$1$, which is obtained by a pure memoryless strategy.
It follows that randomized memoryless strategies are sufficient in MDPs
with mean-payoff objectives to realize the finite-memory value.


\begin{lemma} \label{lem:randomized-memoryless}
In all MDPs with a generalized mean-payoff objective, there exists 
an optimal-for-finite-memory strategy that is randomized memoryless.
\end{lemma}

\paragraph{Polynomial-time algorithm}
We present a polynomial-time algorithm to compute the value
in  generalized mean-payoff MDPs with finite-memory strategies.
The key steps of the algorithm are:
\begin{compactitem}
\item The algorithm determines all end-components with value~$1$ (the 
winning end-components), and then computes the maximum probability to reach 
the union of the winning end-components (see Remark~\ref{rem:key}). 

\item The first step to obtain the winning end-components is to define a 
linear program based on the frequencies that
gives a union of end-components with frequencies
that satisfy the generalized mean-payoff objective.
However, this union of end-components itself may not be connected,
even though it is part of a larger end-component.
In the infinite-memory strategy case, the paths between the union of
end-components can be used with vanishing frequency to ensure
the generalized mean-payoff objectives.
However, for finite-memory strategies connectedness of the union of the 
end-components must be ensured.
We show how to combine the linear program with a graph-based algorithm
to ensure connectedness and get a polynomial-time algorithm.

\end{compactitem}

\smallskip\noindent{\em Frequency-based linear program.}
It is known that the winning probability for reachability objectives can be computed in polynomial time using a reduction
to linear programming~\cite{FV97}. To complete the proof, we present a solution 
to compute the winning end-components in polynomial time.
Our approach extends a technique for finding in a graph a cycle with sum 
of rewards equal to zero in all dimensions~\cite{KS88}.
First, we present a linear program $\LP$ to find a union of end-components with nonnegative
sum of rewards (the end-components may be disjoint). 
The variables $f_{s,a}$ represent the frequency of playing action $a$ in state $s$. 
The linear program $\LP$ consists of the following constraints (see also \figurename~\ref{fig:frequency}):

\renewcommand{\theenumi}{(\arabic{enumi})}
\begin{enumerate}[(E1)]
\item for each $s \in S$: $\sum_{a \in A} f_{s,a} = \sum_{t \in S} \sum_{a \in A} f_{t,a} \cdot \PP(t,a)(s)$ \label{n1} \\[-6pt]

\item $\sum_{s \in S} \sum_{a \in A} f_{s,a} \cdot \rwd(s) \geq 0$ (component-wise) \label{n2} \\[-6pt]

\item $\sum_{s \in S} \sum_{a \in A} f_{s,a} = 1$ \label{n3} \\[-6pt]

\item for each $s \in S$ and $a \in A$: $f_{s,a} \geq 0$ \label{n4} 

\end{enumerate}
\renewcommand{\theenumi}{\arabic{enumi}}
The equations~(E\ref{n1}) above express that in every state, the incoming frequency is
equal to the outgoing frequency. Equation~(E\ref{n2}) ensures that the mean-payoff
value is nonnegative (in all dimensions). Equations~(E\ref{n3}) and~(E\ref{n4}) 
require that the frequencies are nonnegative and sum up to~$1$.

\smallskip\noindent{\em Illustration.}
In the example of \figurename~\ref{fig:frequency}, a solution to the linear program 
gives for instance $f_1 = \frac{1}{16}$ and $f_3 = \frac{7}{16}$, which corresponds
to a randomized memoryless strategy that chooses from $s_1$ to go to $s_2$ with 
probability $\frac{1}{1+7} = \frac{1}{8}$ and to go to $\{s_3,s_4\}$ with probability
$\frac{7}{1+7} = \frac{7}{8}$. This strategy satisfies the conjunction of mean-payoff objectives 
with probability~1 (it ensures that the long-run average of the rewards is $\frac{1}{32} \geq 0$
in both dimensions).

\smallskip\noindent{\em Issues regarding connectedness.}
Arguments similar to the proof of~\cite[Theorem 2.2]{KS88} show that the linear
program $\LP$ has a solution if and only if there exists a union of end-components in $M$
and associated frequencies with nonnegative sum of rewards. However, this union
of end-components need not to be connected and thus may not be an end-component
(see \figurename~\ref{fig:frequency-not-connected} where the union of 
the end-components $\{s_1\}$ and $\{s_3\}$ corresponds to a solution of $\LP$).
Note that connectedness is not an issue for infinite-memory strategies: 
in the example of \figurename~\ref{fig:frequency-not-connected} there exists 
an infinite-memory strategy to ensure the mean-payoff objectives with probability~1 
(see~\cite[Lemma~7]{VCDHRR15}).

\smallskip\noindent{\em Ensuring connectedness and frequencies.}
To find single end-components with nonnegative sum of rewards, we adapt a 
technique presented in~\cite[Section~3]{KS88}. Construct a graph $G_M$ with set $S$ of vertices,
and for each pair $(s,a) \in S \times A$, if the linear program $\LP \land f_{s,a} > 0$ 
has a solution, add edges $(s,t)$ in $G_M$ for all $a$-successors $t$ of $s$.   
If the graph $G_M$ is strongly connected, then it defines an end-component with nonnegative
sum of rewards in $M$. Otherwise, consider the maximum-scc decomposition of $G_M$, 
and iterate the algorithm in each scc, until the state space reduces to one element. 
The algorithm identifies in this way all (maximal) winning end-components and
arguments similar to~\cite[Theorem~3.3]{KS88} show that this algorithm runs in polynomial
time, as the recursion depth is bounded by the number of states, and the scc 
decomposition ensures that the graphs in each recursive call of a given depth 
are disjoint.


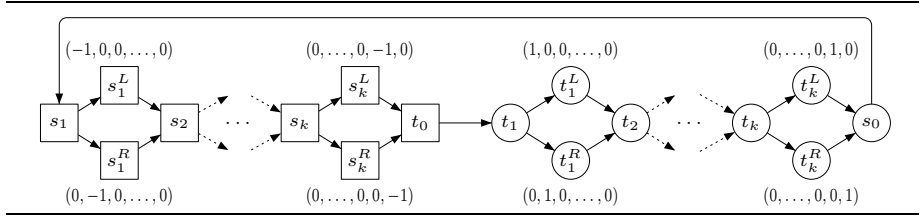
\begin{figure*}[!t]
  \centering
  \hrule  height .33pt
    {\scriptsize 
\begin{picture}(115,28)(0,2)

\gasset{Nw=6,Nh=6,Nmr=3,rdist=1, loopdiam=6}
\gasset{Nw=5,Nh=5,Nmr=2.5,rdist=1, loopdiam=6, linewidth=0.12}

\node[Nmarks=n, Nmr=0](n1)(5,14){$s_1$}
\node[Nmarks=n, Nmr=0](n1L)(13,19){$s^L_1$}
\nodelabel[ExtNL=y, NLangle=90, NLdist=1](n1L){\sb{$(-1,0,0,\dots,0)$}}
\node[Nmarks=n, Nmr=0](n1R)(13,9){$s^R_1$}
\nodelabel[ExtNL=y, NLangle=270, NLdist=1](n1R){\sb{$(0,-1,0,\dots,0)$}}
\node[Nmarks=n, Nmr=0](n2)(21,14){$s_2$}

\node[Nframe=n, Nmr=0](n2L)(30,19){}
\node[Nframe=n, Nmr=0](n2R)(30,9){}

\drawedge[ELpos=50, ELside=l, curvedepth=0](n1,n1L){}
\drawedge[ELpos=50, ELside=l, curvedepth=0](n1,n1R){}
\drawedge[ELpos=50, ELside=l, curvedepth=0](n1L,n2){}
\drawedge[ELpos=50, ELside=l, curvedepth=0](n1R,n2){}
\drawedge[ELpos=50, ELside=l, curvedepth=0, dash={.3 .5}0](n2,n2L){}
\drawedge[ELpos=50, ELside=l, curvedepth=0, dash={.3 .5}0](n2,n2R){}

\node[Nframe=n](dots)(29,14){$\dots$}

\node[Nframe=n, Nmr=0](n2L)(28,19){}
\node[Nframe=n, Nmr=0](n2R)(28,9){}

\node[Nmarks=n, Nmr=0](n2)(37,14){$s_k$}
\drawedge[ELpos=50, ELside=l, curvedepth=0, dash={.3 .5}0](n2L,n2){}
\drawedge[ELpos=50, ELside=l, curvedepth=0, dash={.3 .5}0](n2R,n2){}

\node[Nmarks=n, Nmr=0](n2L)(45,19){$s^L_k$}
\nodelabel[ExtNL=y, NLangle=90, NLdist=1](n2L){\sb{$(0,\dots,0,-1,0)$}}
\node[Nmarks=n, Nmr=0](n2R)(45,9){$s^R_k$}
\nodelabel[ExtNL=y, NLangle=270, NLdist=1](n2R){\sb{$(0,\dots,0,0,-1)$}}
\node[Nmarks=n, Nmr=0](n3)(53,14){$t_0$}

\drawedge[ELpos=50, ELside=l, curvedepth=0](n2,n2L){}
\drawedge[ELpos=50, ELside=l, curvedepth=0](n2,n2R){}
\drawedge[ELpos=50, ELside=l, curvedepth=0](n2L,n3){}
\drawedge[ELpos=50, ELside=l, curvedepth=0](n2R,n3){}

\node[Nmarks=n](t1)(65,14){$t_1$}
\node[Nmarks=n](t1L)(73,19){$t^L_1$}
\nodelabel[ExtNL=y, NLangle=90, NLdist=1](t1L){\sb{$(1,0,0,\dots,0)$}}
\node[Nmarks=n](t1R)(73,9){$t^R_1$}
\nodelabel[ExtNL=y, NLangle=270, NLdist=1](t1R){\sb{$(0,1,0,\dots,0)$}}
\node[Nmarks=n](t2)(81,14){$t_2$}

\node[Nframe=n](t2L)(89,19){}
\node[Nframe=n](t2R)(89,9){}

\drawedge[ELpos=50, ELside=l, curvedepth=0](n3,t1){}
\drawedge[ELpos=50, ELside=l, curvedepth=0](t1,t1L){}
\drawedge[ELpos=50, ELside=l, curvedepth=0](t1,t1R){}
\drawedge[ELpos=50, ELside=l, curvedepth=0](t1L,t2){}
\drawedge[ELpos=50, ELside=l, curvedepth=0](t1R,t2){}
\drawedge[ELpos=50, ELside=l, curvedepth=0, dash={.3 .5}0](t2,t2L){}
\drawedge[ELpos=50, ELside=l, curvedepth=0, dash={.3 .5}0](t2,t2R){}

\node[Nframe=n](dots)(89,14){$\dots$}

\node[Nframe=n](t2L)(88,19){}
\node[Nframe=n](t2R)(88,9){}

\node[Nmarks=n](t2)(97,14){$t_k$}
\drawedge[ELpos=50, ELside=l, curvedepth=0, dash={.3 .5}0](t2L,t2){}
\drawedge[ELpos=50, ELside=l, curvedepth=0, dash={.3 .5}0](t2R,t2){}

\node[Nmarks=n](t2L)(105,19){$t^L_k$}
\nodelabel[ExtNL=y, NLangle=90, NLdist=1](t2L){\sb{$(0,\dots,0,1,0)$}}
\node[Nmarks=n](t2R)(105,9){$t^R_k$}
\nodelabel[ExtNL=y, NLangle=270, NLdist=1](t2R){\sb{$(0,\dots,0,0,1)$}}
\node[Nmarks=n](t3)(113,14){$s_0$}

\drawedge[ELpos=50, ELside=l, curvedepth=0](t2,t2L){}
\drawedge[ELpos=50, ELside=l, curvedepth=0](t2,t2R){}
\drawedge[ELpos=50, ELside=l, curvedepth=0](t2L,t3){}
\drawedge[ELpos=50, ELside=l, curvedepth=0](t2R,t3){}

\drawline[AHnb=1,arcradius=1](113,16.5)(113,28)(5,28)(5,16.5)




\end{picture}
}
  \hrule  height .33pt
    \caption{A family of generalized mean-payoff games where player~$\ma$ (round states) needs 
	exponential memory to win almost-surely (and finite memory is sufficient).\newline \label{fig:exponential}} 
\end{figure*}

\begin{theorem}\label{theo:mdp} 
The following assertions hold for MDPs with generalized mean-payoff objectives
$\Omega \in \{\MeanSup, \MeanInf\}$:
\begin{enumerate}

\item There exists a randomized memoryless strategy $\straa$ 
such that 
$\vaf(\Omega,s) = \prob{s}^{\straa}(\Omega,s)$ for all states $s$ 
(i.e., randomized memoryless optimal strategies wrt. to finite-memory strategies).

\item The value and value-strategy problems for generalized mean-payoff MDPs 
under finite-memory strategies (i.e., whether $\vaf(\Omega,s) \geq \lambda$)
can be solved in polynomial time.
\end{enumerate}
\end{theorem}

\noindent{\em Insufficiency of pure memoryless strategies.}
While we show that randomized memoryless strategies are sufficient, 
the example of \figurename~\ref{fig:frequency} shows that pure memoryless
strategies are not sufficient to achieve the optimal finite-memory value:
from $s_1$, a pure memoryless strategy can either choose $s_2$ and then the
mean-payoff value in the first dimension is $-\frac{3}{2} < 0$, or choose $\{s_3,s_4\}$
and then the mean-payoff value in the second dimension is $-\frac{1}{2} < 0$. 
Thus for all pure memoryless strategies, the generalized mean-payoff objective is
violated with probability~1 although there exists an almost-sure
winning \emph{randomized} memoryless strategy (see the paragraph {\em Illustration} 
after Lemma~\ref{lem:randomized-memoryless}).

\subsection{Generalized mean-payoff objectives under finite-memory in 2\half-player games}

We present a result analogous to Theorem~\ref{theo:half-memoryless} 
for generalized mean-payoff stochastic games showing that memoryless strategies are sufficient
for player~$\mi$ against finite-memory strategies. 
Note that the result extends Theorem~\ref{theo:half-memoryless}
as mean-payoff-sup objectives are not closed under shuffling 
(Remark~\ref{rmk:mean-payoff-inf-not-closed-under-shuffling}).

\begin{theorem}\label{theo:half-memoryless-mean-payoff}
In stochastic games with objective $\Omega \in \{\MeanSup, \MeanInf\}$, 
there exists an optimal-for-finite-memory strategy 
for player~$\ma$, there exists a memoryless optimal-for-finite-memory strategy
for player~$\mi$, and determinacy holds under finite-memory strategies, that is for all states $s$:
$$\vaf(\Omega,s) = \vbf(\Omega,s) =: v(s), \text{and}$$
$$ \sup_{\straa \in \Straa^{FM}} \inf_{\strab \in \Strab} \prob{s}^{\straa,\strab}(\Omega,s) = v(s) = 
\inf_{\strab \in \Strab^{PM}} \sup_{\straa \in \Straa^{FM}} \prob{s}^{\straa,\strab}(\Omega,s).
$$
\end{theorem}

\begin{proof}
For mean-payoff-inf objectives ($\Omega = \MeanInf\}$) the result follows
from Theorem~\ref{theo:half-memoryless}.
We consider mean-payoff-sup objectives ($\Omega = \MeanSup\}$).
Once a finite-memory strategy for player~$\ma$ is fixed we have an MDP
for player~$\mi$, with a disjunction of mean-payoff objectives.
We now analyze the MDP problem.
Since the objective is prefix-independent, by Remark~\ref{rem:key}, 
every end-component has value either~1 or~0. 
It follows that in every end-component with value~1, one of the mean-payoff 
objectives is satisfied with value~1, for which positional strategies are 
sufficient~\cite{FV97,Puterman}.
The optimal reachability to winning end-components is also achieved by
positional strategies. 
Given the existence of positional strategies for player~$\mi$ in MDPs, we
consider the game problem.
It follows that once a finite-memory strategy for player~$\ma$ is fixed, the 
counter-strategy for player~$\mi$ is also finite-memory, and for finite-memory
strategies mean-payoff-sup and mean-payoff-inf objectives coincide.
The desired result follows.
\end{proof}


It follows that the value problem for generalized mean-payoff games
with finite-memory strategies can be solved in coNP by guessing a memoryless
strategy for player~$\mi$ and checking whether the value of the resulting
MDP under finite-memory strategies for player~$\ma$ is above the given threshold, 
which can be done in polynomial time (Theorem~\ref{theo:mdp}).
By the result of~\cite[Lemma~5, Lemma~6]{VCDHRR15}, the problem of deciding
the existence of a finite-memory almost-sure winning strategy for player~$\ma$ in a game 
(even with only deterministic transitions) with a conjunction of mean-payoff-sup 
or mean-payoff-inf objectives is coNP-hard. 
Theorem~\ref{theo:finite-coNP-complete} summarizes the results of this section.

\begin{theorem}\label{theo:finite-coNP-complete}
The value and value-strategy problems for stochastic games
with generalized mean-payoff-(inf or sup) objectives 
played with finite-memory strategies for player~$\ma$ 
(and finite- or infinite-memory strategies for player~$\mi$) 
are coNP-complete.
\end{theorem}

\subsection{Memory bounds for strategies in 2\half-player games}\label{sec:membou}
We present both exponential lower bound and upper bound on memory of strategies.

\smallskip\noindent{\bf Lower bound.}
We show that in games where finite memory is sufficient to win almost-surely 
a conjunction of mean-payoff objectives, exponential memory is necessary in general,
even with randomized strategies. There exists a 
family of generalized mean-payoff games in which player~$\ma$ has a finite-memory almost-sure winning 
strategy, and every almost-sure winning strategy needs exponential memory.
The family of games is illustrated in \figurename~\ref{fig:exponential},
and is essentially the same family used in the proof of~\cite[Lemma 8]{CRR14}
(for 2-player games with pure finite-memory strategies). 

\smallskip\noindent{\em Lower bound family.}
For $k=1,2,\dots$, the $k$-th game in the family consists of $2k$ gadgets where
the first $k$ gadgets belong to player~$\mi$, and the last $k$ gadgets belong to 
player~$\ma$, for a total of $6k + 2$ states (\figurename~\ref{fig:exponential}). 
The action set is $\{L,R\}$ and the unique $L$-successor of $s_i$ is $s_i^L$,
the unique $L$-successor of $t_i$ is $t_i^L$, and similarly for $R$-successors.
Thus all transitions in the games are deterministic.
The reward function has dimension $2k$. All states have reward $0$ in all dimensions 
except the following states and dimensions:
\begin{compactitem}
\item[$\bullet$] $\rwd_{2i-1}(s_i^L) = \rwd_{2i}(s_i^R) = -1$
\item[$\bullet$] $\rwd_{2i-1}(t_i^L) = \rwd_{2i}(t_i^R) = 1$
\end{compactitem}
An almost-sure winning strategy for player~$\ma$ in this game is to copy in 
every state $t_i$ the choice of player~$\mi$ in state $s_i$, namely choosing 
$t_i^L$ if player~$\mi$ has chosen $s_i^L$, and choosing $t_i^R$ if player~$\mi$ 
has chosen $s_i^R$. This strategy ensures that the sum of rewards in all 
dimensions is bounded (in the interval $[-1,1]$) and thus the mean-payoff value
is $0$ in all dimensions. The memory needed to describe this strategy requires
$k$ bits to remember the last $k$ choices of player~$\mi$, thus a memory set of 
size $2^{k}$.

\begin{lemma}\label{lem:exponential}
There exists a family of games $G_k$ with $O(k)$ states and generalized mean-payoff
objective of dimension $2k$ such that a finite-memory almost-sure winning strategy exists,
and all almost-sure winning randomized strategies require memory of size at least $2^k$.
\end{lemma}

\begin{proof}
We show that no randomized strategy with memory of size less than $2^{k}$ 
is almost-sure winning in the family of games $G_k$ presented above (see \figurename~\ref{fig:exponential}). 
The proof is by contradiction. Consider a strategy $\straa$
with memory of size less than $2^{k}$, and show that $\straa$ is not 
almost-sure winning. Since $\straa$ has memory of size less than $2^{k}$,
after every visit to $s_1$, it is possible for player~$\mi$ to make two
different sequences of choices until reaching $t_1$, such that from $t_1$
the strategy $\straa$ of player~$\ma$ behaves identically because it is in the same memory state. 
Formally, for all play prefixes $\rho \in S(S^{4k+2})^*$ with $\Last(\rho) = s_0$, 
there exist two play prefixes $\rho_1, \rho_2 \in S^{2k+1}$ with $\rho_1 \neq \rho_2$
such that $\hat{\straa}_u(m_0, \rho \rho_1) = \hat{\straa}_u(m_0, \rho \rho_2)$,
where $m_0$ is the initial memory and $\hat{\straa}_u$ is the update function
of the strategy $\straa$.
We can view $\rho_1$ and $\rho_2$ as sequences of $k$ bits taking value $L$
or $R$. Since $\rho_1 \neq \rho_2$, there is an index $1 \leq i \leq k$ where
the $i$-th bit differs in $\rho_1$ and $\rho_2$; we say that there is a star at
position $i$ after prefix $\rho$. Intuitively, the star represents the possibility
for player~$\mi$ to pick either $L$ or $R$ at position $i$ (by playing according to
$\rho_1$ or $\rho_2$) without affecting the future choices of player~$\ma$.
Note that the choice of $L$ or $R$ at position $i$ changes the reward in 
dimensions $2i-1$ and $2i$ (respective rewards $0$ and $-1$ or $-1$ and $0$) and
possibly in other dimensions.
Consider the (finite-memory) strategy for player~$\mi$ that for each such 
prefix $\rho$ plays as prescribed by $\rho_1$, 
and consider the Markov chain obtained by fixing the strategies $\straa$
and $\strab$ in the game $G_k$. Consider a closed recurrent set in this Markov chain,
and compute for each index $i=1,\dots,k$ the frequency of occurrence of a star
(the frequencies are well-defined because the Markov chain has finitely many states).
It is easy to see that for some index $i$ this frequency is greater than $0$ (in fact,
at least $\frac{1}{k}$). Let $f_*$ be the frequency of occurrence of a star at 
position $i$,
relative to the frequency of occurrence of state $s_1$.
Given the index $i$, we consider the dimensions $2i -1$ and $2i$. 
In the closed recurrent set, we can analogously compute the relative frequency of the left and right
choices at $s_i$ and $t_i$ (without counting the choices when there is a star at position $i$), 
thus computing the frequency $f_1$ of rewards $(-1,0)$
in dimensions $2i-1$ and $2i$, the frequency $f_2$ of rewards $(0,-1)$, 
the frequency $g_1$ of $(1,0)$, and $g_2$ of $(0,1)$. 
Thus we have $f_1 + f_2 + f_* = 1$ and $g_1 + g_2 = 1$. Since $f_* > 0$,
it follows that either $f_1 + f_* > g_1$ or $f_2 + f_* > g_2$, that is
either in dimension $2i -1$ or $2i$, the expected reward can be made negative 
by choosing either always the choice $L$ or always the choice $R$ at all states that have 
a star at position $i$ (remember that this modifies the strategy $\strab$ of player~$\mi$, 
but does not affect the choices made by the strategy $\straa$ of player~$\ma$ along
the play). It follows that almost-surely the mean-payoff value 
is negative in that dimension, showing that $\straa$ is not almost-surely winning.

It follows that exponential memory is necessary in general to win almost-surely 
a generalized mean-payoff game (even when finite memory is sufficient).
\end{proof}

\smallskip\noindent{\bf Upper bound.} 
Theorem~\ref{theo:membou} and Theorem~\ref{theo:mdp} establish an $\abs{A}^{\abs{S_{\mi}}}$ 
upper bound on memory required for optimal-for-finite-memory strategies. 
Thus we obtain the following result.

\begin{theorem}\label{theo:memgenmean}
The optimal bound for memory required for optimal-for-finite-memory strategies for 
player~$\ma$ in generalized mean-payoff stochastic games is exponential.
\end{theorem}

\section{Generalized Mean-Payoff Objectives under Infinite-Memory Strategies}\label{sec:infmem}

In this section, we consider games with a conjunction of mean-payoff objectives
and infinite-memory strategies for player~$\ma$ (which are more powerful than 
finite-memory strategies~\cite[Lemma~7]{VCDHRR15}).


\subsection{$\MeanInf$ objectives}

Since $\MeanInf$ objectives are prefix-independent and closed under shuffling, 
it follows from the results of~\cite[Theorem~5.2]{GK14}
that for player~$\mi$ memoryless optimal strategies exist. Therefore the
value and value-strategy problems can be solved in coNP by guessing a (optimal)
memoryless strategy for player~$\mi$, and then solving an MDP with conjunction
of mean-payoff objectives under infinite-memory strategies, which can be done in 
polynomial time by the result of~\cite[Section~3.2]{BBCFK14}.
A matching coNP-hardness bound is known for 2-player games~\cite[Theorem~7]{VCDHRR15}.

\begin{theorem}\label{theo:mean-payoff-inf-infinite}
The value and the value-strategy problems for stochastic games 
with generalized mean-payoff-inf objectives under infinite-memory 
strategies 
are coNP-complete.
\end{theorem}

\subsection{$\MeanSup$ objectives}

We focus on the \emph{almost-sure winning problem} for generalized mean-payoff objectives, 
which is to decide whether there exists an almost-sure winning strategy 
for player~$\ma$ from a given state.
We show that the almost-sure winning problem is in NP~$\cap$~coNP for a conjunction
of mean-payoff-sup objectives.

\begin{remark}\label{rem:red-almost-sure}
As mentioned in~\cite[Remark~1]{CDGO14}, it follows from the results of~\cite[Lemma~7]{CHH09} 
and~\cite[Theorem~4.1]{GH10} that since mean-payoff objectives are prefix-independent objectives,
the memory requirement for optimal strategies of both players is the same as 
for almost-sure winning strategies, and if the almost-sure winning problem is in 
NP~$\cap$~coNP, then the value-strategy problem is also in 
NP~$\cap$~coNP. Thus it will follow from our results
for the almost-sure problem 
that the value and value-strategy problems are 
in NP~$\cap$~coNP for $\MeanSup$ objectives.
\end{remark}

For mean-payoff-sup objectives, we show that the almost-sure winning problem
is in NP~$\cap$~coNP. For player~$\ma$ to be almost-sure winning for a conjunction
of mean-payoff-sup objectives, it is necessary to be almost-sure winning for
each one-dimensional mean-payoff-sup objective, and we show that it is sufficient.
An almost-sure winning strategy is to play in rounds according to the almost-sure winning
strategy of each one-dimensional objective successively, for a duration that
is always finite but long enough to ensure the corresponding one-dimensional 
average of rewards (thus over finite plays) tends to the objective mean-payoff value
with high probability (that tends to $1$ as the number of rounds increases).

\begin{lemma}\label{lem:all-one}
If in a game, for every one-dimensional mean-payoff-sup objective $\MeanSup_j$ ($j=1,\dots,k$)
all states are almost-sure winning for player~$\ma$,
then for the objective $\MeanSup = \bigwedge_{1 \leq j  \leq k} \MeanSup_j$ 
all states are almost-sure winning for player~$\ma$.
\end{lemma}

\begin{proof}
To show this result, consider for each objective $\MeanSup_j$ an almost-sure winning 
strategy $\straa_j$. Then for all $\epsilon > 0$, there exists a number of steps 
$N_{\epsilon}$ such that for all $N \geq N_{\epsilon}$,  for all states $s$, 
all dimensions $1 \leq j \leq k$, and all strategies $\strab$
of player~$\mi$, we have~\cite[Lemma~1]{CDGO14}:
$$ \prob{s}^{\straa_j,\strab}\left(\left\{ s_0 s_1 \dots \in S^\omega \mid \frac{1}{N} 
\cdot \sum_{i = 0}^{N-1} \rwd_j(s_{i}) \geq -\epsilon\right\}\right) \geq 1-\epsilon.$$
Thus by playing according to strategy $\straa_j$ for a large enough number of steps, 
the average reward in dimension $j$ can be made arbitrarily close to $0$,
with probability arbitrarily close to $1$.

Let $W = \max_{s \in S, 1 \leq j \leq k} \abs{\rwd_j(s)}$ be the 
largest reward in absolute value. 
We construct an almost-sure winning strategy $\straa$ for the objective $\MeanSup$ 
as follows. The strategy $\straa$ plays in rounds numbered $1,2,\dots$, where
at round $\ell$ the strategy $\straa$ plays for each dimension $j=1,\dots,k$ successively as follows:
let $Z$ be the length of the current play prefix, and let $\epsilon_Z = \frac{1}{Z}$;
play according to strategy $\straa_j$ from the current play prefix 
for $N_Z = \max\{N_{\epsilon_Z},Z^2\cdot W\}$ steps.

Now we show that the strategy $\straa$ is almost-sure winning for the conjunction
of mean-payoff-sup objectives. Consider an arbitrary dimension $j$. For all
play prefixes of length $Z$, the total reward is at least $-Z \cdot W$. 
It follows that after playing according to $\straa_j$ for $N_Z$ steps (at round $\ell$), 
since $N_Z \geq N_{\epsilon_Z}$, against all strategies $\strab$ of player~$\mi$, we have:

$$\prob{s}^{\straa,\strab}\left(\left\{ s_0 s_1 \dots \in S^\omega \mid \frac{1}{N_{Z}} 
\cdot \sum_{i = 0}^{N_{Z}-1} \rwd_j(s_{i}) \geq \frac{-Z\cdot W - \epsilon_Z \cdot N_{Z}}{Z+N_{Z}} \right\}\right) 
\geq 1-\epsilon_Z.
$$
Note that $\frac{-Z \cdot W - \epsilon_Z \cdot N_{Z}}{Z+N_{Z}} \geq 
\frac{-Z \cdot W - \epsilon_Z \cdot N_{Z}}{N_{Z}} \geq
\frac{-Z \cdot W}{Z^2\cdot W} - \epsilon_Z \geq
- \frac{2}{Z}$.
Therefore,
$$\limsup_{Z \to \infty} \prob{s}^{\straa,\strab} \left(\left\{ s_0 s_1 \dots \in S^\omega \mid \frac{1}{N_{Z}} 
\cdot \sum_{i = 0}^{N_{Z}-1} \rwd_j(s_{i}) \geq \frac{-2}{Z} \right\}\right)
\geq  \limsup_{Z \to \infty} 1-\epsilon_Z = 1.
$$
By Fatou's lemma~\cite{Billingsley}, for a sequence $\cale_Z$ of measurable sets we have that 
$\limsup_{Z \to\infty}\prob{}(\cale_Z) \leq \prob{}(\limsup_{Z\to\infty} \cale_Z)$.
Hence we have 

$$\prob{s}^{\straa,\strab} \underbrace{ \left( \limsup_{Z \to \infty} \left\{ s_0 s_1 \dots \in S^\omega \mid \frac{1}{N_{Z}} 
\cdot \sum_{i = 0}^{N_{Z}-1} \rwd_j(s_{i}) \geq \frac{-2}{Z} \right\} \right)}_{\varphi} = 1.$$

We show that $\varphi \subseteq \MeanSup_j$. Consider a play $\rho = s_0 s_1 \dots \in \varphi$.
Since for $Z \to \infty$, we have $N_Z \to \infty$ and $\frac{-2}{Z} \to 0$, 
for \emph{all} $\epsilon > 0$ there exist infinitely many integers $N$ such that 
$\frac{1}{N} \cdot \sum_{i = 0}^{N-1} \rwd_j(s_{i}) \geq -\epsilon$.
Therefore $\limsup_{N \to\infty} \frac{1}{N} \cdot \sum_{i = 0}^{N-1} \rwd_j(s_{i}) \geq 0$
and thus $\rho \in \MeanSup_j$. It follows that 
$\prob{s}^{\straa,\strab} \left( \MeanSup_j \right) = 1$ and since this holds 
for all dimensions $j$, we have $\prob{s}^{\straa,\strab} \left( \MeanSup \right) = 1$.
Hence $\straa$ is almost-sure winning for the conjunction of mean-payoff-sup 
objectives, from all states~$s$.
\end{proof}

\begin{theorem}\label{theo:mean-payoff-sup-infinite}
The value and the value-strategy problems for stochastic games 
with generalized mean-payoff-sup objectives under infinite-memory 
strategies are in NP~$\cap$~coNP.
\end{theorem}

\smallskip\noindent{\bf The NP algorithm.}
By Lemma~\ref{lem:all-one}, an NP algorithm for the almost-sure winning
problem is to guess the set $\Win$ of almost-sure winning states, and check
that $(i)$ $\Win$ \emph{induces a subgame} for player~$\ma$, that is for every
state $s \in \Win \cap S_{\ma}$ there exists an action $a \in A$ such that
$\Supp(\PP(s,a)) \subseteq \Win$, and for every state $s \in \Win \cap S_{\mi}$ 
for all actions $a \in A$, we have $\Supp(\PP(s,a)) \subseteq \Win$;
and $(ii)$ in the subgame induced by $\Win$ every state is almost-sure winning
for every one-dimensional mean-payoff-sup objective, which can be checked in NP~\cite{LigLip69}.
This establishes the first part of Theorem~\ref{theo:mean-payoff-sup-infinite}.

\smallskip\noindent{\bf The coNP algorithm.}
We now show that the almost-sure winning problem is also in coNP. 
Given a set $T \subseteq S$ of states, let 
 
\begin{align*}
\cpremi(T) =\, & \{s \in S_{\ma} \mid \forall a \in A: \Supp(\PP(s,a)) \cap T \neq \emptyset \}\, \cup \\
             & \{s \in S_{\mi} \mid \exists a \in A: \Supp(\PP(s,a)) \cap T \neq \emptyset \}.
\end{align*}
be the set of \emph{controllable predecessors} of set $T$ for player~$\mi$.
The \emph{positive attractor} $\atmi(T)$ of $T$ for player~$\mi$ is the least 
fixed point of the operator $f: 2^S \to 2^S: X \mapsto T \cup \cpremi(X)$,
that is the
set of all states from which player~$\mi$ has a memoryless strategy to ensure that $T$ is reached
with positive probability, against all strategies of player~$\ma$.
Note that the set $S \setminus \atmi(\{s\})$ induces a subgame for player~$\ma$.

Let $\overline{\Win} = S \setminus \Win$ be the set of all states 
that are not almost-sure winning for player~$\ma$ (for the conjunction
$\MeanSup$ of mean-payoff-sup objectives). 
By Lemma~\ref{lem:all-one} it follows that $(i)$ there exist a dimension $j$ and 
a state $s \in \overline{\Win}$ that is not almost-sure winning for player~$\ma$ 
(for the one-dimensional objective $\MeanSup_j$), and $(ii)$ in the subgame induced by the set 
$S \setminus \atmi(\{s\})$, the set of states that are not almost-sure winning 
for player~$\ma$ (for $\MeanSup$) is $\overline{\Win} \setminus \atmi(\{s\})$
(because if player~$\ma$ has an almost-sure winning strategy in the subgame, 
then this strategy is also almost-sure winning in the original game).

It follows that the set $\overline{\Win}$ can be partitioned into sets $U_1, \dots, U_n$
such that for all $1 \leq i \leq n$ there exists $\emptyset \neq R_i \subseteq U_i$ such that 
$U_i = \atmi(R_i)$ and for some dimension $j$ 
the states of $R_i$ are not almost-sure winning for player~$\ma$ for the one-dimensional objective $\MeanSup_j$,
in the subgame induced by $S \setminus (U_1 \cup \dots \cup U_{i-1})$.
In each set $U_i$ we can fix a memoryless strategy $\strab$ for player~$\mi$ as follows:
in $R_i$ where the game objective is $\MeanSup_j$, fix a memoryless optimal 
strategy for player~$\mi$ to violate the objective $\MeanSup_j$ with positive probability 
(which exists for one-dimensional mean-payoff games~\cite{LigLip69}), 
and in $U_i \setminus R_i$ fix a memoryless strategy to ensure $R_i$ is reached with 
positive probability (which exists by the definition of positive attractor).
The strategy $\strab$ ensures that against any strategy of player~$\ma$, the
objective $\MeanSup$ is violated with positive probability. 

Hence a coNP algorithm for the almost-sure winning problem is to guess
the set $\overline{\Win}$ of states that are not almost-sure winning for player~$\ma$,
and a memoryless strategy $\strab$ for player~$\mi$. The verification can be done 
by checking that the MDP obtained by playing $\strab$ in the game
is not almost-sure winning for the $\MeanSup$ objective, which can be done 
in polynomial time~\cite[Section~3.2]{BBCFK14}.
This establishes the second part of Theorem~\ref{theo:mean-payoff-sup-infinite}.

Note that improving the  NP~$\cap$~coNP bound to PTIME for even single dimensional objectives
would be a major breakthrough, as it would imply a polynomial solution for
simple stochastic games~\cite{Condon92}. 

\section{Conclusion}\label{sec:con}
In this work we consider 2\half-player games with generalized mean-payoff objectives.
We establish an optimal complexity result of coNP-completeness under finite-memory 
strategies, 
which significantly improves the previously known semi-decision procedure, 
even for the special case of the almost-sure problem.
We also establish optimal bounds for the memory required for finite-memory strategies.
Given several quantitative objectives, a more general problem is to consider
a different probability threshold for each objective (in contrast we consider
the probability of the conjunction of the objectives).
For the almost-sure problem the more general problem coincides with the problem
we consider. 
The more general problem is open, even for the special case of  
multiple reachability objectives in 2\half-player games.

\bibliographystyle{abbrv}
\bibliography{biblio}





\end{document}